\documentclass[11pt]{article}
\pdfoutput=1
\usepackage{latexsym}
\usepackage{amsmath}
\usepackage{amssymb}
\usepackage{amsfonts,amsthm}
\usepackage{authblk}
\usepackage[top=1in, bottom=1in, left=1in, right=1in]{geometry}
\usepackage{thmtools, thm-restate}

\usepackage{tikz}
\usetikzlibrary{arrows,decorations.pathmorphing,decorations.shapes,backgrounds,fit}
\pgfkeys{tikz/cc1/.style={dotted} }
\pgfkeys{tikz/cc4/.style={dashed} }
\pgfkeys{tikz/cc3/.style={decorate,decoration=bumps} }
\pgfkeys{tikz/cc2/.style={decorate,decoration=snake} }
\pgfkeys{tikz/solid/.style={line width=2pt} }
\pgfkeys{tikz/solid2/.style={line width=1pt} }
\pgfkeys{tikz/nodostile1/.style={draw,line width=1pt}}
\pgfkeys{tikz/nodostile2/.style={draw,line width=2pt}}
\pgfkeys{tikz/arcostile1/.style={very thick} }
\pgfkeys{tikz/arcostile2/.style={red,very thick} }
\pgfkeys{tikz/arcostile3/.style={thick} }
\pgfkeys{tikz/arcostile4/.style={line width=2pt} }

\newenvironment{mylist}[1]{
\setbox1=\hbox{#1}
\begin{list}{}{
\setlength{\labelwidth}{\wd1}
\setlength{\leftmargin}{\wd1}
\addtolength{\leftmargin}{0em}
\addtolength{\leftmargin}{\labelsep}
\setlength{\rightmargin}{1em}}}{\end{list}}

\newcommand{\litem}[1]{\item[#1\hfill]}

\newcommand{\ignore}[1]{}

\newtheorem{theorem}{Theorem}[section]
\newtheorem{lemma}[theorem]{Lemma}

\newtheorem{corollary}[theorem]{Corollary}

\newtheorem{property}[theorem]{Property}

\usepackage{array}
\newcolumntype{C}[1]{>{\centering\let\newline\\\arraybackslash\hspace{0pt}}m{#1}}

\usepackage{thmtools}
\usepackage{thm-restate}

\usepackage{hyperref}

\usepackage{cleveref}

\usepackage{microtype}

\usepackage{latexsym}
\usepackage{multirow}
\usepackage{amsmath}
\usepackage{amssymb}
\usepackage{amsfonts,amsthm}
\usepackage{adjustbox}
\usepackage{mathcomp}

\usepackage{tabu}
\usepackage{longtable}

\usepackage{dcolumn}
\newcolumntype{d}[1]{D{.}{.}{#1}}

\usepackage{array}
\newcolumntype{C}[1]{>{\centering\let\newline\\\arraybackslash\hspace{0pt}}m{#1}}

\usepackage[lined, algoruled, linesnumbered]{algorithm2e}

\newcommand{\compress}[1]{}

\title{{\bf Incremental $2$-Edge-Connectivity in Directed Graphs}
}

\author[1]{Loukas Georgiadis}
\author[2]{Giuseppe F.~Italiano}
\author[2]{Nikos Parotsidis}
\affil[1]{University of Ioannina, Greece. \texttt{loukas@cs.uoi.gr}}
\affil[2]{Universit\`a di Roma ``Tor Vergata'', Italy.
\texttt{giuseppe.italiano@uniroma2.it, nikos.parotsidis@uniroma2.it}}

\begin{document}

\maketitle

\begin{abstract}
In this paper, we initiate the study of the dynamic maintenance of $2$-edge-connectivity relationships in directed graphs.
We present an algorithm that can update the $2$-edge-connected blocks of a directed graph with $n$ vertices through a sequence of
$m$ edge insertions in a total of $O(mn)$ time. After each insertion, we can answer the following queries in asymptotically optimal time:
\begin{itemize}
\item Test in constant time if two query vertices $v$ and $w$ are $2$-edge-connected. Moreover, if  $v$ and $w$ are not $2$-edge-connected, we can produce in constant time a ``witness'' of this property, by exhibiting
an edge that is contained in all paths from $v$ to $w$ or in all paths from $w$ to $v$.
\item Report in $O(n)$ time all the $2$-edge-connected blocks of $G$.
\end{itemize}
To the best of our knowledge, this is the first  dynamic
algorithm for $2$-connectivity problems on directed graphs, and it matches the best known bounds for simpler problems, such as incremental transitive closure.
\end{abstract}

\section{Introduction}

The design of dynamic graph algorithms is one of the classic areas in
theoretical computer science. In this
setting, the input of a graph problem is being changed via a sequence of updates, such as edge insertions and deletions.
A dynamic graph algorithm aims at updating efficiently the solution
of a problem after each update, faster than recomputing it from scratch. A dynamic graph problem is said to be \emph{fully dynamic} if the update operations
include both insertions and deletions of edges, and it
is said to be \emph{partially dynamic} if only one type of update, either insertions or deletions,
is allowed. More specifically, a dynamic graph problem is said to be \emph{incremental} (resp., \emph{decremental}) if only
insertions (resp., deletions) are allowed.

In this paper, we present new incremental algorithms for $2$-edge connectivity problems on directed graphs (digraphs). Before defining the problem, we first review some definitions.
Let $G=(V,E)$ be a digraph. $G$ is \emph{strongly connected} if there is a directed path from each vertex to every other vertex.
The \emph{strongly connected components} (in short \emph{SCCs}) of $G$ are its maximal strongly connected subgraphs.
Two vertices $u,v \in V$  are \emph{strongly connected}  if they belong to the same SCC
of $G$.
An edge of $G$ is a \emph{strong bridge} if its removal increases the number of SCCs.
Let $G$ be strongly connected: $G$ is \emph{$2$-edge-connected} if it has no strong bridges.
The \emph{$2$-edge-connected components} of $G$ are its maximal $2$-edge-connected subgraphs. Two vertices $u, v\in V$ are said to be \emph{$2$-edge-connected}, denoted by  $u \leftrightarrow_{\mathrm{2e}} v$, if there are two edge-disjoint directed paths from $u$ to $v$  and two edge-disjoint directed paths from $v$ to $u$. (Note that a path from $u$ to $v$ and a path from $v$ to $u$ need not be edge-disjoint).
A \emph{$2$-edge-connected block} of a digraph $G=(V,E)$ is defined as a maximal subset $B \subseteq V$ such that $u \leftrightarrow_{\mathrm{2e}} v$ for all $u, v \in B$.
Figure \ref{figure:2ECB-example} illustrates the $2$-edge-connected blocks of a digraph.

We remark that in digraphs
$2$-vertex and $2$-edge connectivity have a much richer and more complicated structure than in undirected graphs. To see this, observe that, while in undirected graphs blocks are exactly the same as components, in digraphs there is a substantial difference between those two notions.
In particular, the edge-disjoint paths that make two vertices $2$-edge-connected in a block might use
vertices that are outside of that block, while in a component those paths must lie  completely inside that component. In other words,
two vertices that are $2$-edge-connected (and thus in the same $2$-edge-connected block) may lie in different $2$-edge-connected components (e.g., vertices $i$ and $j$ in Figure \ref{figure:2ECB-example}, each of them being in a $2$-edge-connected component by itself).
As a result, $2$-connectivity problems on digraphs appear to be much harder than on undirected graphs. For undirected graphs
it has been known for over 40 years how to compute the $2$-edge- and $2$-vertex- connected components in linear time~\cite{dfs:t}.
In the case of digraphs, however, only $O(mn)$ algorithms were known
(see e.g.,~\cite{2vcb:jaberi,2VCC:Jaberi2015,makino,nagamochi}).
It was shown only recently how to compute the $2$-edge- and $2$-vertex- connected blocks in linear time \cite{2ECB,2VCB}, and the best current bound for computing the $2$-edge- and the $2$-vertex- connected components
is $O(n^2)$~\cite{2CC:HenzingerKL14}.

\begin{figure}[t!]
\begin{center}
\includegraphics[trim={0 0 0 6cm}, clip=true, width=1.0\textwidth]{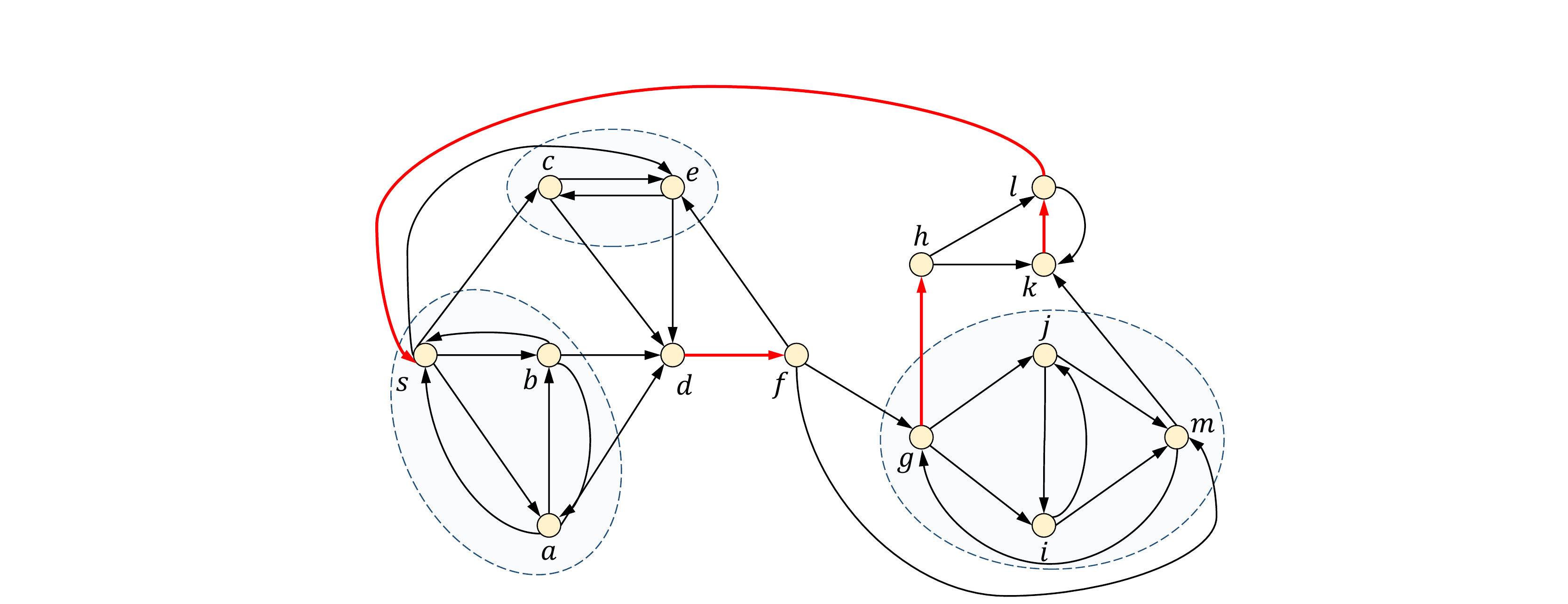}	
\caption{The $2$-edge-connected blocks of a digraph $G$. Strong bridges of $G$ are shown red. (Better viewed in color.)}
\label{figure:2ECB-example}
\end{center}
\end{figure}


\paragraph{Our Results.}
In this paper, we initiate the study of the dynamic maintenance of $2$-edge-connectivity relationships in directed graphs.
We present an algorithm that can update the $2$-edge-connected blocks of a digraph $G$ with $n$ vertices through a sequence of
$m$ edge insertions in a total of $O(mn)$ time. After each insertion, we can answer the following queries in asymptotically optimal time:
\begin{itemize}
\item Test in constant time if two query vertices $v$ and $w$ are $2$-edge-connected.  Moreover, if $v$ and $w$ are not $2$-edge-connected, we can produce in constant time a ``witness'' of this property, by exhibiting
an edge that is contained in all paths from $v$ to $w$ or in all paths from $w$ to $v$.
\item Report in $O(n)$ time all the $2$-edge-connected blocks of $G$.
\end{itemize}
Ours is the first dynamic
algorithm for $2$-connectivity problems on digraphs, and it matches the best known bounds for simpler problems, such as incremental transitive closure~\cite{Ita86}. This is a substantial improvement  over the $O(m^2)$ simple-minded algorithm, which recomputes the $2$-edge-connected blocks from scratch after each edge insertion.


\paragraph{Related Work.}
Many efficient algorithms for several dynamic graph problems have been proposed in the literature, including dynamic connectivity~\cite{HK99,HLT01,PT07,Thorup2000}, minimum spanning trees~\cite{EGIN97,F85,HK01,HLT01}, edge/vertex connectivity~\cite{EGIN97,HLT01} on undirected graphs, and transitive closure~\cite{DI08,HK95,King99} and shortest paths~\cite{DI04,King99,Thorup04} on digraphs.
Once again, dynamic problems on digraphs appear to be harder than on undirected graphs. Indeed, most of the dynamic algorithms on undirected graphs have polylog update bounds, while dynamic algorithms on digraphs have higher polynomial update bounds.  The hardness of dynamic algorithms on digraphs has been recently supported also by conditional lower bounds~\cite{AW14}.

\paragraph{Our Techniques.}
Known algorithms for computing the $2$-edge-connected blocks of a digraph $G$~\cite{2ECB,2C:GIP:arXiv} hinge on properties that seem very difficult to dynamize. The algorithm in~\cite{2ECB} uses very complicated data structures based on 2-level auxiliary graphs. The loop nesting forests used in~\cite{2C:GIP:arXiv} depends heavily on an underlying dfs tree of the digraph, and
the incremental maintenance of dfs trees on general digraphs is still an open problem (incremental algorithms are known only for the special case of DAGs~\cite{FGN97}).
Despite those inherent difficulties, we find a way to bypass loop nesting forests by suitably combining the
approaches in~\cite{2ECB,2C:GIP:arXiv} in a novel framework, which is
amenable to dynamic implementations.
Another complication is that, although our problem is incremental, strong bridges may not only be deleted but also added (when a new SCC is formed). As a result, our data structures undergo a fully dynamic repertoire of updates, which is known to be harder.
By organizing carefully those updates, we are still able to obtain the desired bounds.

\section{Dominator trees and $2$-edge-connected blocks}
\label{sec:definitions}

We assume the reader is familiar with standard graph terminology, as contained
for instance in~\cite{clrs}.
Given a rooted tree, we denote by $T(v)$ the subtree of $T$ rooted at $v$ (we also view $T(v)$ as the set of descendants of $v$).
Given a digraph $G=(V,E)$, and a set of vertices $S \subseteq V$, we denote by $G[S]$ the subgraph induced by $S$. We introduce next some of the building blocks of our new incremental algorithm.

\subsection{Flow graphs, dominators, and bridges}
\label{sec:dominators}

A \emph{flow graph} is a digraph with a distinguished \emph{start vertex} $s$ such that every vertex is reachable from $s$. Let $G=(V,E)$ be a strongly connected graph.
The \emph{reverse digraph} of $G$, denoted by $G^R=(V, E^R)$, is obtained
by reversing the direction of all edges.
Let $s$ be a fixed but arbitrary start vertex of a strongly connected digraph $G$.
Since $G$ is strongly connected, all vertices are reachable from $s$ and reach $s$, so we can view both $G$ and $G^R$ as flow graphs with start vertex $s$. To avoid ambiguities, throughout the paper we will denote those flow graphs respectively by $G_s$ and $G_s^R$.
Vertex $u$ is a \emph{dominator} of vertex $v$ ($u$ \emph{dominates} $v$) in $G_s$ if every path from $s$ to $v$ in $G_s$ contains $u$.
We let \emph{Dom$(v)$} denote be the set of dominators of $v$.
The dominator relation can be represented by a tree $D$ rooted at $s$,  the \emph{dominator tree} of $G_s$: $u$ dominates $v$ if and only if $u$ is an ancestor of $v$ in $D$.
For any $v \not= s$,  we denote by $d(v)$ the parent of $v$ in $D$.
Similarly, we can define the dominator relation in the flow graph $G_s^R$, and let $D^R$ denote the dominator tree of $G_s^R$, and $d^R(v)$ the parent of $v$ in $D^R$.
Lengauer and Tarjan~\cite{domin:lt} presented an algorithm for computing dominators in  $O(m \alpha(m,n))$ time for a flow graph with $n$ vertices and $m$ edges, where $\alpha$ is a functional inverse of Ackermann's function~\cite{dsu:tarjan}.
Subsequently, several linear-time algorithms
were discovered~\cite{domin:ahlt,dominators:bgkrtw,dominators:Fraczak2013,Gabow:Poset:TALG}.
An edge $(u,v)$ is a \emph{bridge} of a flow graph $G_s$ if all paths from $s$ to $v$ include $(u,v)$.\footnote{Throughout the paper, to avoid confusion we use consistently the term  \emph{bridge} to refer to a bridge of a flow graph and the term \emph{strong bridge} to refer to a strong bridge in the original graph.}
Let $s$ be an arbitrary start vertex of $G$.
The following properties were proved in \cite{Italiano2012}.

\begin{property}
\label{property:strong-bridge} \emph{(\cite{Italiano2012})}
Let $s$ be an arbitrary start vertex of $G$. An edge $e=(u,v)$ is strong bridge of $G$ if and only if it is a bridge of $G_s$ (so $u=d(v)$) or a bridge of $G_s^R$ (so $v=d^R(u)$) or both.
\end{property}

\begin{figure}[t!]
\begin{center}
\includegraphics[trim={0 0 0 10cm}, clip=true, width=1.0\textwidth]{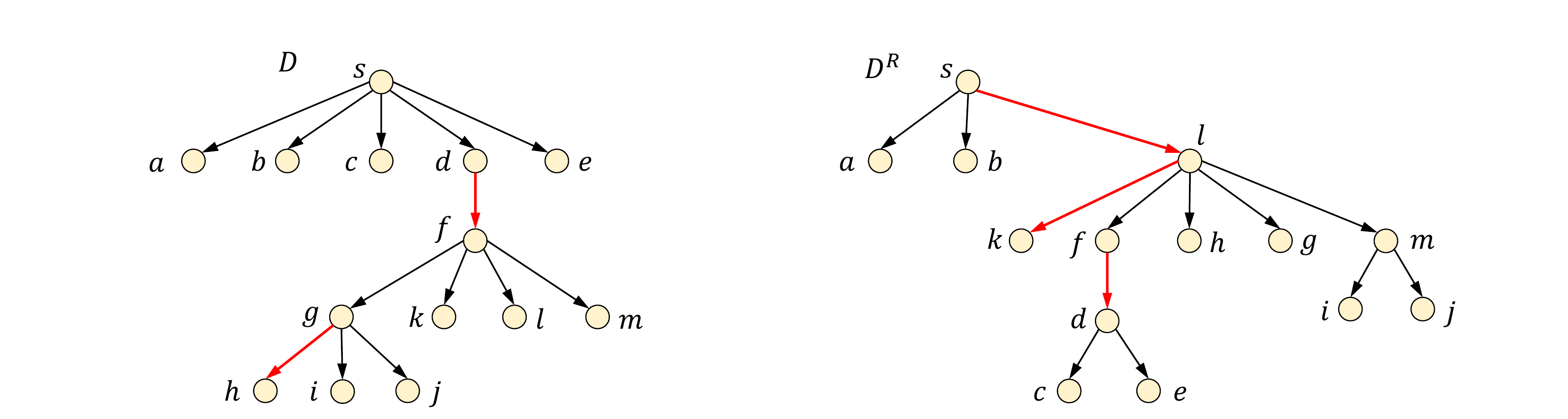}	
\caption{The dominator trees of flow graphs $G_s$ and $G_s^R$. Strong bridges of $G$ are shown red. (Better viewed in color.)}
\label{figure:DomTrees}
\end{center}
\end{figure}

As a consequence, of Property~\ref{property:strong-bridge},
all the strong bridges of the digraph
$G$ can be obtained from the bridges of the flow graphs $G_s$ and $G_s^R$, and thus there can be at most $2(n-1)$ strong bridges overall.
Figure~\ref{figure:DomTrees} illustrates the dominator trees $D$ and $D^R$ of the flow graphs $G_s$ and $G_s^R$ that correspond to the strongly connected digraph $G$ of Figure \ref{figure:2ECB-example}.
After deleting from the dominator trees $D$ and $D^R$ respectively the bridges of $G_s$ and $G_s^R$, we obtain the \emph{bridge decomposition}  of $D$ and $D^R$ into forests $\mathcal{D}$ and $\mathcal{D}^R$.
Throughout the paper, we denote by $D_u$ (resp., $D_u^R$) the tree in $\mathcal{D}$ (resp., $\mathcal{D}^R$) containing vertex $u$, and by $r_u$ (resp., $r^R_u$) the root of $D_u$ (resp., $D_u^R$).
The following lemma from \cite{2ECB} holds for a flow graph $G_s$ of a strongly connected digraph $G$ (and hence also for the flow graph $G_s^R$ of $G^R$).

\begin{lemma}
\label{lemma:partition-paths} \emph{(\cite{2ECB})}
Let $G$ be a strongly connected digraph and let $(u,v)$ be a strong bridge of $G$. Also, let $D$
be the dominator tree of the
flow graph $G_s$, for an arbitrary start vertex $s$.
Suppose $u=d(v)$. Let $w$ be any vertex that is not a descendant of $v$ in $D$. Then there is path from $w$ to $v$ in $G$ that does not contain any proper descendant of $v$ in $D$. Moreover, all simple paths in $G$ from $w$ to any descendant of $v$ in $D$ must contain the edge $(d(v),v)$.
\end{lemma}

\subsection{Loop nesting forests and bridge-dominated components}
\label{sec:loop-nesting}

Let $G$ be a digraph.
A \emph{loop nesting forest} represents a hierarchy of strongly connected subgraphs of $G$~\cite{st:t}, 
defined with respect to a dfs tree $T$ of $G$, as follows.
For any vertex $u$, the \emph{loop} of $u$, denoted by
$\mathit{loop}(u)$ is the set of all descendants $x$ of $u$ in $T$ such that there is a path from $x$ to $u$ in $G$ containing only descendants of $u$ in $T$.
Any two vertices in $\mathit{loop}(u)$ reach each other. Therefore, $\mathit{loop}(u)$ induces a strongly connected subgraph of $G$; it is the unique maximal set of descendants of $u$ in $T$ that does so.
The $\mathit{loop}(u)$ sets form a laminar family of subsets of $V$:
for any two vertices $u$ and $v$, $\mathit{loop}(u)$ and $\mathit{loop}(v)$ are either disjoint or nested.
The \emph{loop nesting forest} $H$ of $G$, with respect to $T$, is the forest in which the parent of any vertex $v$, denoted by $h(v)$, is the nearest proper ancestor $u$ of $v$ in $T$ such that $v \in \mathit{loop}(u)$ if there is such a vertex $u$, and null otherwise.
Then $\emph{loop}(u)$ is the set of all descendants of vertex $u$ in $H$, which we will also denote as $H(u)$ (the subtree of $H$ rooted at vertex $u$).
A loop nesting forest can be computed in linear time~\cite{dominators:bgkrtw,st:t}.
Since we deal with strongly connected digraphs, each vertex is contained in a loop, so $H$ is a tree.
Therefore, we will refer to $H$ as the \emph{loop nesting tree} of $G$.
Let $e=(u,v)$ be a bridge of the flow graph $G_s$, and
let $G[D(v)]$ denote the subgraph induced by the vertices in $D(v)$.
Let $C$ be an SCC of $G[D(v)]$: we say that $C$ is an
\emph{$e$-dominated component} of $G$.
We also say that $C \subseteq V$ is a \emph{bridge-dominated component} if it is an $e$-dominated component for some bridge $e$:
As shown in the following lemma,
bridge-dominated components form a laminar family, i.e., any two components in $\mathcal{C}$ are either disjoint or one contains the other.

\begin{lemma}
\label{lemma:laminar}
Let $\mathcal{C}$ be the set family of all bridge-dominated components. Then $\mathcal{C}$ is laminar.
\end{lemma}
\begin{proof}
Let $C$ and $C'$ be two different sets of $\mathcal{C}$ such that $C \cap C' \not= \emptyset$.
Let $e$ and $e'$ be the nearest bridge ancestors of $C$ and $C'$, respectively, such that $C$ is $e$-dominated and $C'$ is $e'$-dominated.
Since $C$ and $C'$ contain a common vertex, say $v$, we can assume, without loss of generality, that $e$ is an ancestor of $e'$.
Any two vertices, $w \in C$ and $z \in C'$, are strongly connected, since they are both strongly connected with $v$.
Moreover, both $w$ and $z$ are descendants of $e$, and hence it must be $C' \subset C$.
This implies that any two components in $\mathcal{C}$ are either disjoint or one contains the other, which yields the lemma.
\end{proof}

Let $e=(u,v)$ be a bridge of $G_s$, and let $w$ be a vertex in $D(v)$ such that $h(w) \not\in D(v)$. As shown in \cite{2C:GIP:arXiv},  $H(w)$ induces an SCC
in $G[D(v)]$, and thus it is an $e$-dominated component.

\subsection{Bridge decomposition and auxiliary graphs}
\label{sec:bridge-decomposition}
Now we define a notion of \emph{auxiliary graphs} that play a key role in our approach.
Auxiliary graphs were defined in \cite{2ECB} to
decompose the input digraph $G$ into smaller digraphs (not necessarily subgraphs of $G$) that maintain the original $2$-edge-connected blocks of $G$.
Unfortunately, the auxiliary graphs of \cite{2ECB} are not suitable
for our purposes, and we need a slightly different definition.
For each root $r$ of a tree in the bridge decomposition $\mathcal{D}$ we define the \emph{auxiliary graph $\widehat{G}_r = (V_r, E_r)$ of $r$} as follows.
The vertex set
$V_r$ of $\widehat{G}_r$ consists of all the vertices in $D_r$.
The edge set $E_r$ contains all the edges of $G$ among the vertices of $V_r$, referred to as
\emph{ordinary} edges, and a set of \emph{auxiliary} vertices, which are obtained by contracting vertices in $V\setminus V_r$, as follows.
Let $v$ be a vertex in $V_r$ that has a child $w$ in $V \setminus V_r$. Note that $w$ is a root in the bridge decomposition $\mathcal{D}$ of $D$.
For each such child $w$ of $v$, we contract $w$ and all its descendants in $D$ into $v$.
Figure \ref{figure:AuxiliaryGraphs} shows the  bridge decomposition of the dominator tree $D$ and the corresponding auxiliary graphs.
The  bridge decomposition of $D^R$ is shown in
Figure \ref{figure:AuxiliaryGraphsReverse}. 
Differently from \cite{2ECB}, our auxiliary graphs do not preserve the $2$-edge-connected blocks of $G$.
Note that each vertex appears exactly in one auxiliary graph. Furthermore, each original edge corresponds to at most one auxiliary edge.
Therefore, the total number of vertices in all auxiliary graphs is $n$, and the total number of edges is at most $m$.
We use the term \emph{auxiliary components} to refer to the SCCs
of the auxiliary graphs.

\begin{figure}
\begin{center}
\includegraphics[trim={0 0 0 4cm}, clip=true, width=1.0\textwidth]{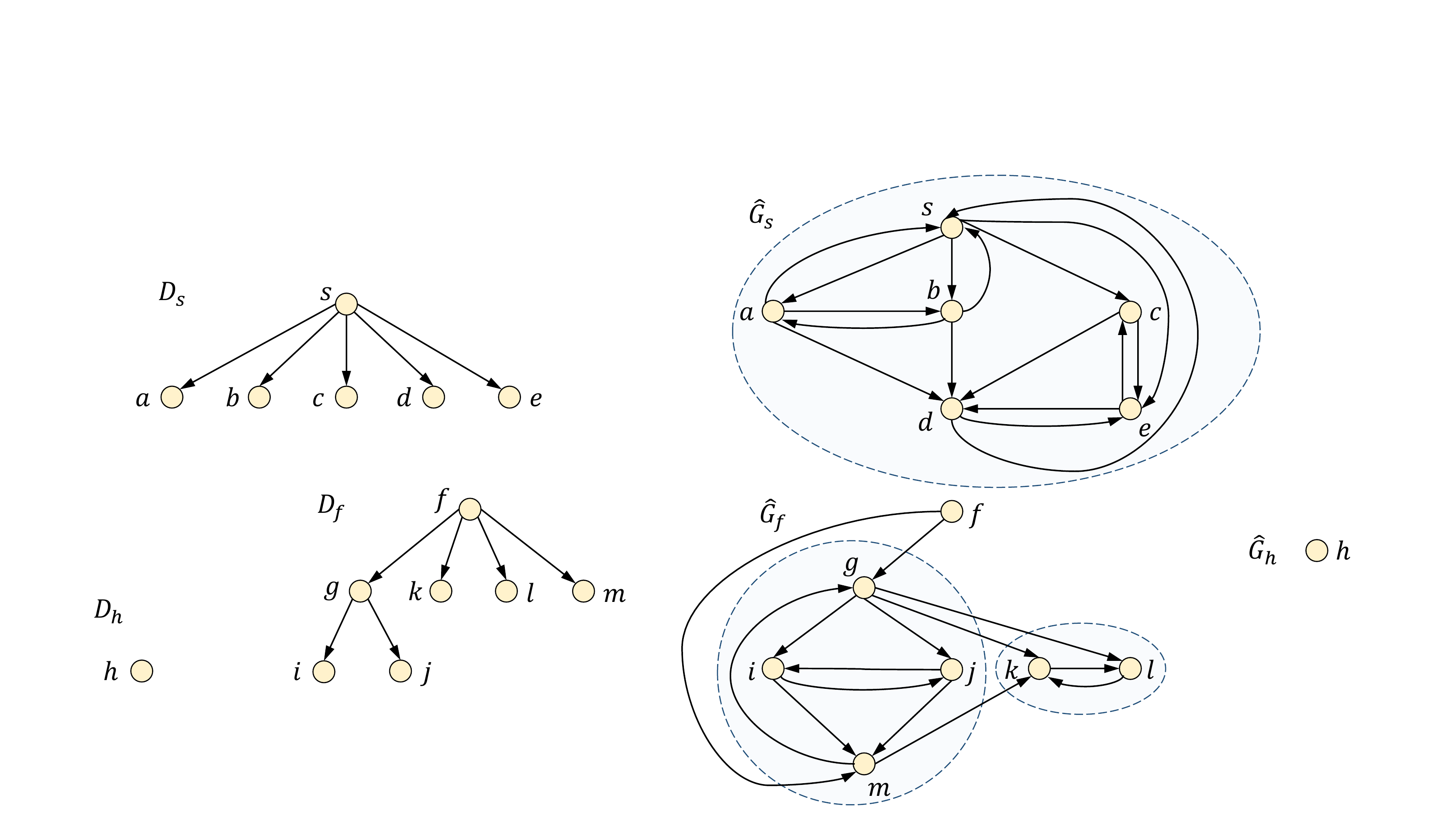}	
\caption{The bridge decomposition of the dominator tree $D$ of Figure \ref{figure:DomTrees}, the corresponding auxiliary graphs $\widehat{G}_r$ and their SCCs
shown encircled.}
\label{figure:AuxiliaryGraphs}
\end{center}
\end{figure}

\begin{figure}
\begin{center}
\includegraphics[trim={0 0 0 2cm}, clip=true, width=1.0\textwidth]{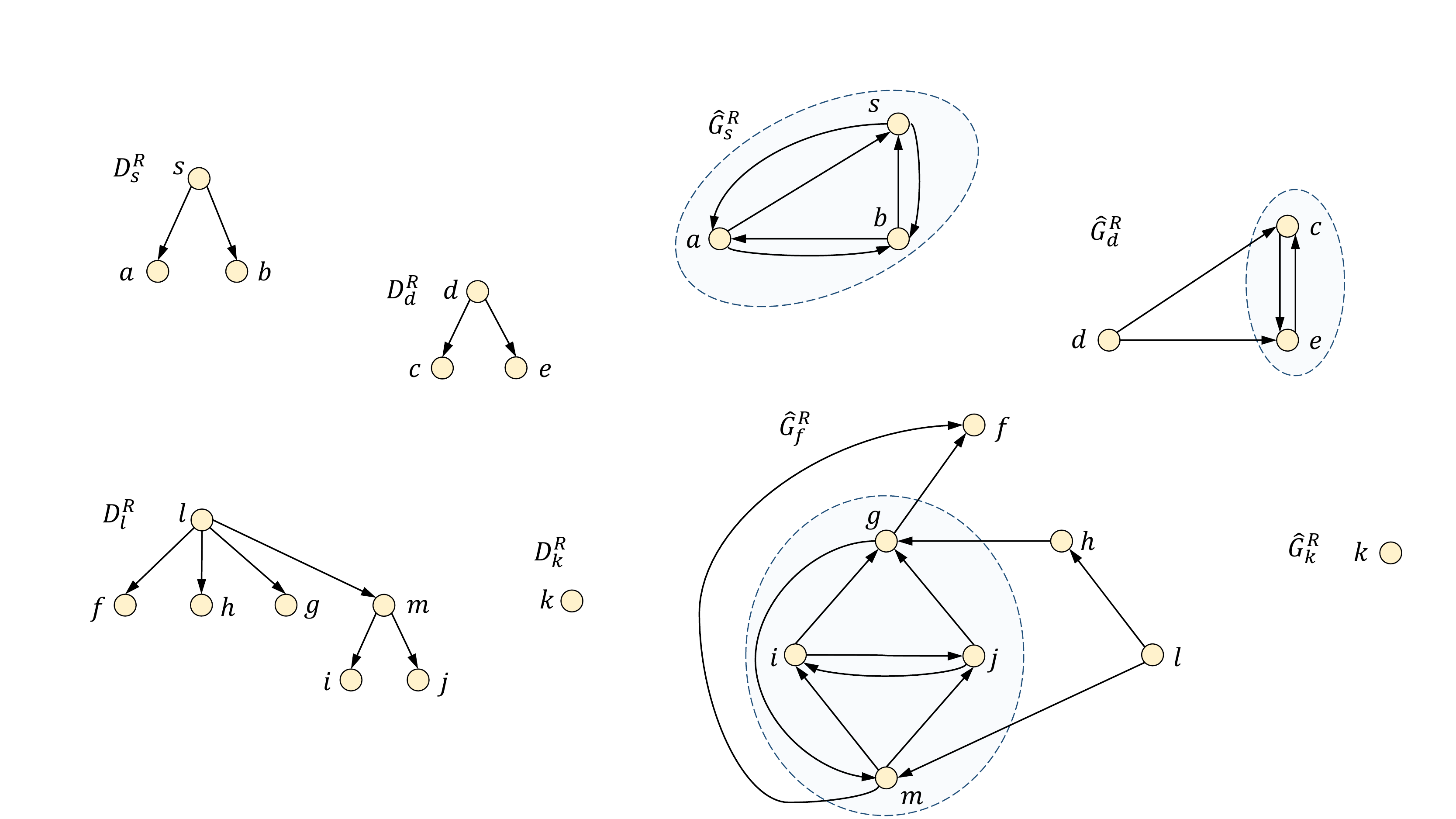}	
\caption{The bridge decomposition of the dominator tree $D^R$ of Figure \ref{figure:DomTrees}, the corresponding auxiliary graphs $\widehat{G}^R_r$ and their SCCs shown encircled.}
\label{figure:AuxiliaryGraphsReverse}
\end{center}
\end{figure}

\begin{lemma}
\label{lemma:auxiliary-graphs}
All the auxiliary graphs of a flow graph $G_s$ can be computed in linear time.
\end{lemma}
\begin{proof}
To construct the auxiliary graph $\widehat{G}_r = (V_r, E_r)$ we need to specify how to compute the shortcut edge of each edge entering $V_r$ from a descendant of $r$ in $V \setminus V_r$.
To do this efficiently we need to test ancestor-descendant relations in $D$. We can test this relation in constant time using a preorder numbering of $D$~\cite{domin:tarjan}.
Let $\mathit{pre}(v)$ denote the preorder number of a vertex $v$ in $D$.
Suppose $(u,v)$ is an edge such that $v \in V_r$, $u \not \in V_r$, and $u$ is a descendant of $r$ in $D$.
We need to find the nearest ancestor $u'$ of $u$ in $D$ such that $u' \in V_r$. Then, $(u,v)$ corresponds to the shortcut edge $(u',v)$.
To compute the shortcut edges of $\widehat{G}_r$, we create a list $B_r$ that contains the edges $(u,v)$ such that $u \in V_r$,
$u \not \in V_r$, and $u$ is a descendant of $r$ in $D$.
For each such edge $(u,v)$ we need to find the nearest ancestor $u'$ of $u$ in $D$ such that $u' \in V_r$. Then, $(u,v)$ corresponds to the shortcut edge $(u',v)$.
We create a second list $B'_r$ that contains the vertices in $V_r$ that have a child that is not in $V_r$, and sort $B'_r$ in increasing preorder.
Then $u'$ is the last vertex in the sorted list $B'_r$ such that $\mathit{pre}(u') \le \mathit{pre}(u)$.
Thus the shortcut edges can be computed by
bucket sorting and merging.
In order to do these computations in linear time for all auxiliary graphs, we sort all the lists at the same time as follows.
First, we create a unified list $B$ containing the triples $\langle r,\mathit{pre}(u),v\rangle$ for each edge $(u,v)$ that corresponds to a shortcut edge
in the auxiliary graph $\widehat{G}_r$. Next we sort $B$ in increasing order of the first two elements. We also create a second list $B'$ with pairs $\langle r,\mathit{pre}(u')\rangle$, where $u'$
is a vertex in $V_r$ that has a child that is not in $V_r$, and sort the pairs in increasing order.
Finally, we compute the shortcut edges of each auxiliary graph $\widehat{G}_r$ by merging the sorted sublists of $B$ and $B'$ that correspond to the same root $r$.
Then, the shortcut edge for the triple $\langle r,\mathit{pre}(u),v\rangle$ is $(u',v)$, where $\langle r,\mathit{pre}(u')\rangle$ is the last pair in the sorted sublist of $B'$ with root $r$ such that $\mathit{pre}(u') \le \mathit{pre}(u)$.
\end{proof}

Let $C$ be a set of vertices, and let $(u,v)$ be a bridge of $G_s$. The restriction of $C$ in $D_v$ is the set $C_v = C \cap D_v$.

\begin{lemma}
\label{lemma:auxiliary-scc}
Let $e=(u,v)$ be a bridge of $G_s$, and let $C$ be an $e$-dominated component.
Then, the restriction $C_v$ of $C$ in $D_v$ is an SCC of $\widehat{G}_v$.
\end{lemma}
\begin{proof}
Let $x$ and $y$ be any vertices in $D_v$. It suffices to argue that $x$ and $y$ are in the same $e$-dominated component if and only if they are strongly connected in $\widehat{G}_v$.
First suppose that $x$ and $y$ are in different $e$-dominated components. Then $x$ and $y$ are not strongly connected in the subgraph $G[D(v)]$ that is induced by $D(v)$.
Then, without loss of generality, we can assume that all paths from $x$ to $y$ contain a vertex in $V \setminus D(v)$.  By Lemma \ref{lemma:partition-paths} all paths from $x$ to $y$ contain $(u,v)$.
Hence $x$ and $y$ are not strongly connected in $\widehat{G}_v$.

Now suppose that $x$ and $y$ are in the same  $e$-dominated component. That is, $x$ and $y$ are strongly connected in $G[D(v)]$, so there is a path from $x$ to $y$ and a path from $y$ to $x$
containing only vertices in $D(v)$. Then $\widehat{G}_v$ contains, by construction, a path from $x$ to $y$ and a path from $y$ to $x$. Thus $x$ and $y$ are strongly connected in $\widehat{G}_v$.
\end{proof}

\subsection{A new algorithm for $2$-edge-connected blocks}
\label{sec:labeling-algorithm}
We next sketch a new linear-time algorithm to compute the $2$-edge-connected blocks
of a strongly connected digraph $G$
that combines ideas
from \cite{2ECB} and \cite{2C:GIP:arXiv} and that will be useful for our incremental algorithm. We refer to this algorithm as the \textsf{2ECB labeling algorithm}.
Similarly to the algorithm of
\cite{2C:GIP:arXiv}, our algorithm assigns a label to
each vertex, so that two vertices are $2$-edge-connected if and only if they have the same label.
The labels are defined by the bridge decomposition of the dominator trees and by the auxiliary components, as follows.
Let $\widehat{G}_r$ be an auxiliary graph of $G_s$. We pick a \emph{canonical vertex} for each SCC
$C$ of $\widehat{G}_r$, and denote by
$c_x$ the canonical vertex of the SCC
that contains $x$. We define $c_x^R$ for the SCC's
of the auxiliary graphs of $G_s^R$ analogously.
We define the label of $x$ as $\mathit{label}(x) = \langle r_x, c_x, r^R_x, c^R_x\rangle$.

To prove that the \textsf{2ECB labeling algorithm} is correct, we show that the
labels produced by the algorithm  are essentially identical to the labels of \cite{2C:GIP:arXiv}.
In order to do that, we briefly review the linear-time algorithm of \cite{2C:GIP:arXiv} for computing the $2$-edge-connected blocks of $G$.
The algorithm of \cite{2C:GIP:arXiv} computes, for each vertex $x$, a tuple $\mathit{label}'(x) = \langle r_x, h_x, r^R_x, h^R_x\rangle$, where $r_x$ and $r^R_x$ are exactly as before,
while $h_x$ and $h^R_x$ are defined by the loop nesting trees $H$ and $H^R$ respectively.
We say that a vertex $x$ is a \emph{boundary vertex} in $H$ if $h(x) \not\in D_x$, i.e., when $x$ and its parent in $H$ lie in different trees of the bridge decomposition $\mathcal{D}$. As a special case, we also let $s$ be a boundary vertex of $H$.
The \emph{nearest boundary vertex} of $x$ in $H$, denoted by $h_x$, is the nearest ancestor of $x$ in $H$ that is a boundary vertex in $H$. Hence, if $r_x = s$ then $h_x = s$. Otherwise, $h_x$ is the unique ancestor of $x$ in $H$ such that $h_x \in D_x$ and $h(h_x) \not \in D_x$.
We define the \emph{nearest boundary vertex} of $x$ in $H^R$ similarly.
A vertex $x$ is a \emph{boundary vertex} in $H^R$ if $h^R(x) \not\in D^R_x$, i.e., when $x$ and its parent in $H^R$ lie in different trees of $\mathcal{D}^R$. Again, we let $s$ be a boundary vertex of $H^R$.
Then, the \emph{nearest boundary vertex} of $x$ in $H^R$, denoted by $h^R_x$, is the nearest ancestor of $x$ in $H^R$ that is a boundary vertex in $H^R$.
As shown in \cite{2C:GIP:arXiv}, two vertices $x$ and $y$ are $2$-edge-connected if and only if $\mathit{label}'(x) =\mathit{label}'(y)$.

\begin{lemma}
\label{lemma:labeling-algorithm}
Let $x$ and $y$ be any vertices of $G$. Then, $x$ and $y$ are $2$-edge-connected if and only if
$\mathit{label}(x)=\mathit{label}(y)$.
\end{lemma}
\begin{proof}
Let $\mathit{label}'(v) = \langle r_v, h_v, r^R_v, h^R_v\rangle$ be the labels assigned to by the original labeling algorithm of \cite{2C:GIP:arXiv}.
By \cite{2C:GIP:arXiv}, we have that $x$ and $y$ are $2$-edge-connected if and only if $\mathit{label}'(x)=\mathit{label}'(y)$.
Suppose $r_x = r_y$. We show that $h_x = h_y$ if and only if $c_x = c_y$.
The fact that $r_x = r_y$ implies that $x$ and $y$ are in the same auxiliary graph $\widehat{G}_{r_x}$.
By Lemma \ref{lemma:auxiliary-scc}, $x$ and $y$ are strongly connected in $\widehat{G}_{r_x}$ if and only if $h_x = h_y$.
Hence, $h_x = h_y$ if and only if $c_x = c_y$.
The same argument implies that if $r^R_x = r^R_y$, then $h^R_x = h^R_y$ if and only if $c^R_x = c^R_y$.
We conclude that $x$ and $y$ are $2$-edge-connected if and only if
$\mathit{label}(x)=\mathit{label}(y)$.
\end{proof}

\begin{theorem}
\label{theorem:labeling-algorithm}
The \textsf{2ECB labeling algorithm} computes the $2$-edge-connected blocks of a strongly connected digraph in linear time.
\end{theorem}
\begin{proof}
The correctness of the labeling algorithm follows from Lemma~\ref{lemma:labeling-algorithm}. We now bound the runnng time.
The dominator trees and the bridges can be computed in linear time~\cite{dominators:bgkrtw}.
Also all the auxiliary graphs and the auxiliary components can be computed in linear time by Lemma \ref{lemma:auxiliary-graphs} and \cite{dfs:t}.
Hence, all the required labels can be computed in linear time.
\end{proof}

\subsection{Incremental dominators and incremental SCCs}
\label{sec:incremental-scc}

We will use two other
building blocks for our new algorithm, namely
incremental algorithms for maintaining dominator trees
and SCCs.
As shown in \cite{dyndom:2012}, the dominator tree of a flow graph with $n$ vertices can be maintained in $O(m\min\{n,k\}+kn)$ time during a sequence of $k$ edge insertions, where $m$ is the total number of edges after all insertions.
For maintaining the SCCs
of a digraph incrementally, Bender et al.~\cite{Bender:IncCycleDetection:TALG}
presented an algorithm that can handle the insertion of $m$ edges in a digraph with $n$ vertices in $O(m \min\{m^{1/2},n^{2/3}\})$ time. %
Since we aim at an $O(mn)$ bound, we can maintain the SCCs
with a simpler data structure based on
topological sorting
\cite{MarchettiSpaccamela:TopologicalSorting}, augmented so as to handle cycle contractions, as suggested by~\cite{Haeupler:IncTopOrder:TALG}. We refer to this data structure as \textsf{IncSCC},
and we will use it both for maintaining the SCCs of the input graph, and the auxiliary components (i.e., the SCCs
of the auxiliary graphs).
We maintain the SCCs and a topological order for them.
Each SCC
is represented by a canonical vertex, and
the partition of the vertices into SCCs
is maintained through a
disjoint set union data structure~\cite{dsu:tarjan,setunion:tvl}.
The data structure supports
the operation
$\mathit{unite}(p,q)$, which, given canonical vertices $p$ and $q$, merges the SCCs containing $p$ and $q$ into one new SCC and makes $p$
the canonical vertex of the new SCC. It also supports the query
$\mathit{find}(v)$, which returns the canonical vertex of the SCC containing $v$. Here we use the abbreviation $f(v)$ to stand for $\mathit{find}(v)$.
The topological order is represented by a simple numbering scheme, where each canonical vertex is numbered with an integer in the range $[1,n]$, so that
if $(u,v)$ is an edge of $G$, then either $f(u) = f(v)$ ($u$ and $v$ are in the same SCC) or $f(u)$ is numbered less than $f(v)$ (when $u$ and $v$ are in different SCCs).
With each canonical vertex $p$ we store a list $\mathit{out}(p)$ of edges leaving vertices that are in the same SCC as $p$, i.e., edges $(u,v)$ with $f(u)=p$. Note that $\mathit{out}(p)$ may contain multiple vertices in the same SCC (i.e., vertices $u$ and $v$ with $f(u)=f(v)$), due to the SCC contractions (and shortcut edges, in case of the auxiliary components) during edge insertions. Also,  $\mathit{out}(p)$ may contain loops, that is, vertices $v$ with $\mathit{f}(v)=p$.
Each  $\mathit{out}$ list is stored as a doubly linked circular list, so that we can merge two lists and delete a vertex from a list in $O(1)$.
When the incremental SCC
data structure detects that a new SCC is formed, it locates the SCCs that are merged and chooses a canonical vertex for the new SCC.
The \textsf{IncSCC} data structure can handle $m$ edge insertions in a total of $O(mn)$ time.

\section{Incremental  $2$-edge-connectivity in strongly connected digraphs}
\label{sec:incremental-strongly-connected}

To maintain the $2$-edge-connected blocks of a strongly connected digraph during edge insertions,
we design an incremental version of the labeling algorithm of Section \ref{sec:labeling-algorithm}.
If labels are maintained explicitly, one can answer in $O(1)$ time queries on whether
two vertices are $2$-edge-connected, and report in $O(n)$ time all the $2$-edge-connected blocks (see Section \ref{sec:queries}).
Let $(x,y)$ be the edge to be inserted.
We say that vertex $v$ is \emph{affected} by the update if $d(v)$ (its parent in $D$) changes.
Note that $\mathit{Dom}(v)$ may change even if $v$ is not affected.
Similarly, an auxiliary component (resp., auxiliary graph) is affected if it contains an affected vertex.
We let
$\mathit{nca}(x,y)$
denote the nearest common ancestor of $x$ and $y$ in the dominator tree $D$.
We also denote by $D[u,v]$ the path from vertex $u$ to
vertex $v$ in $D$.
If $\mathit{nca}(x,y)$
and $y$ are in different subtrees in the bridge decomposition of $D$ before the insertion of the edge $(x,y)$, we let $(p,q)$ be the first bridge encountered on the path $D[\mathit{nca}(x,y),y]$ (Figure \ref{figure:BridgeDecompositionNCA}). For any vertex $v$, we denote by $\mathit{depth}(v)$ the depth of $v$ in $D$.

\begin{figure}[t!]
\begin{center}
\includegraphics[trim={0 0 0 7cm}, clip=true, width=1.0\textwidth]{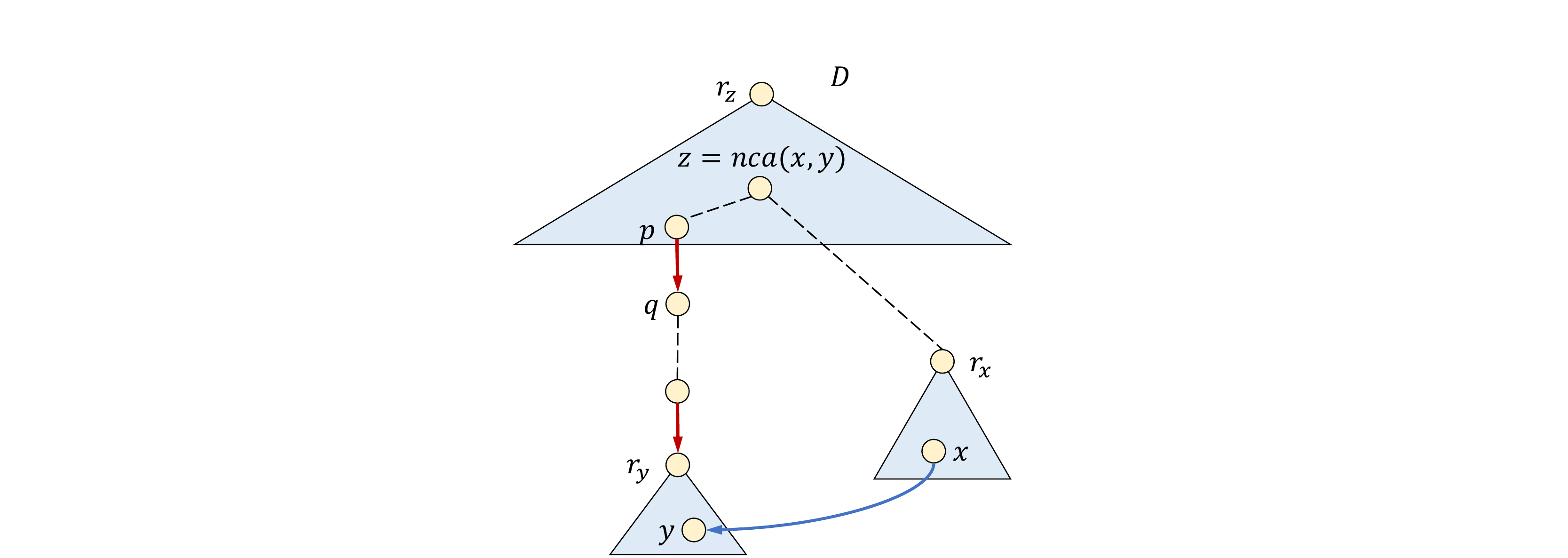}\caption{The bridge decomposition of $D$ before the insertion of a new edge $(x,y)$.}
\label{figure:BridgeDecompositionNCA}
\end{center}
\end{figure}

\subsection{Affected vertices and canceled bridges}
\label{sec:affected}

There are affected vertices after the insertion of $(x,y)$ if and only if $\mathit{nca}(x,y)$ is not a descendant of $d(y)$~\cite{irdom:rr94}.
A characterization of the affected vertices is provided by the following lemma, which is a refinement of a result in~\cite{dynamicdominator:AL}.

\begin{lemma}
\label{lemma:insert-affected} \emph{(\cite{dyndom:2012})}
A vertex $v$ is affected after the insertion of edge $(x,y)$ if and only if $\mathit{depth}(\mathit{nca}(x,y)) < \mathit{depth}(d(v))$
and there is a path $\pi$ in $G$
from $y$ to $v$ such that $\mathit{depth}(d(v)) < \mathit{depth}(w)$ for all $w \in \pi$.
If $v$ is affected, then it becomes a child of $\mathit{nca}(x,y)$
in $D$.
\end{lemma}

The algorithm in \cite{dyndom:2012} applies Lemma \ref{lemma:insert-affected} to identify the
affected vertices by starting a search from $y$ (if $y$ is not affected, then no other vertex is).
We assume that the outgoing and incoming edges of each vertex are maintained as
linked lists, so that a new edge can be inserted in $O(1)$, and that
the dominator tree $D$
is represented by the parent function $d$. We also maintain the depth of
vertices in $D$.
We say that a vertex $v$
is \emph{scanned}, if the edges leaving $v$ are examined during the search for affected vertices, and that it is \emph{visited} if there is a scanned vertex $u$ such
that $(u,v)$ is an edge in $G$. Every scanned vertex is either affected or a descendant of an affected vertex in $D$.
By Lemma \ref{lemma:insert-affected}, a visited vertex $v$ is scanned if $\mathit{depth}(\mathit{nca}(x,y)) < \mathit{depth}(d(v))$.
Let $(u,v)$ be a bridge of $G_s$.
We say that $(u,v)$ is \emph{canceled} by the insertion of edge $(x,y)$ if it is no longer a bridge after the insertion.
We say that a bridge
$(u,v)$ is \emph{locally canceled} if $(u,v)$ is a canceled bridge and $v$ is not affected.
Note that if $(u,v)$ is locally canceled, then $u=\mathit{nca}(x,y)$.
In the next lemmata, we consider the effect of the insertion of edge $(x,y)$ on the bridges of $G_s$, and relate the affected and scanned vertices with the auxiliary components.
Recall that
$(p,q)$ is the first bridge encountered on the path $D[\mathit{nca}(x,y),y]$ (Figure \ref{figure:BridgeDecompositionNCA}), $D(v)$ denotes the descendants of $v$ in $D$, and $G[C]$ is the subgraph induced by the vertices in $C$.

\begin{lemma}
\label{lemma:affected-reachable}
Suppose that bridge $(p,q)$ is not locally canceled after the insertion of $(x,y)$.
Let $z=\mathit{nca}(x,y)$ and let $v$ be an affected vertex such that $r_v \not= r_z$. All vertices reachable from $v$ in $G[D(q)]$ are either affected or scanned.
\end{lemma}
\begin{proof}
Let $\pi$ be a path in $G[D(q)]$ from $v$ to another vertex $w$.
Since $(p,q)$ is not locally canceled we have that $z \not= p$ or $w \not= q$.
Hence, all vertices $u$ on $\pi$ satisfy $\mathit{depth}(z) < \mathit{depth}(d(u))$.
We prove the lemma by induction on the number of the edges in $\pi$ that are not edges of the dominator tree $D$.
Suppose that $\pi$ consists of only one edge, i.e., $\pi=(v,w)$. If $w$ is a child of $v$ in $D$, then $w$ is scanned.
Otherwise, the parent property of $D$ \cite{DomCert:TALG} implies that $d(w)$ is an ancestor of $v$. In this case, Lemma \ref{lemma:insert-affected} implies that $w$ is affected. Thus, the induction base holds.
Assume by induction that the lemma holds for any vertex that is reachable in $G[D(q)]$ from an affected vertex through 
a path of at most $k$ edges that are not edges of $D$.
Let $\pi$ be a path from $v$ to $w$. Let $(u,w')$ be the first edge on $\pi$ such that $w'$ is not a descendant of $v$ in $D$. The parent property of $D$ implies again that $d(w')$ is an ancestor of $u$ in $D$. Since $w'$ is not a descendant of $v$, $d(w')$ is also an ancestor of $v$ in $D$. Hence, Lemma \ref{lemma:insert-affected} implies that $w'$ is affected.
The part of $\pi$ from $w'$ to $w$ satisfies the induction hypothesis, and the lemma follows.
\end{proof}

\begin{lemma}
\label{lemma:canceled-bridge}
Let $e=(u,v)$ be a bridge of $G_s$ that is canceled by the insertion of edge $(x,y)$. Then (i) $y$ is a descendant of $v$ in $D$, and (ii) $y$ is in the same $e$-dominated component
as $v$.
\end{lemma}
\begin{proof}
Since $e=(u,v)$ is not a bridge in $G_s'$, there must be a path $\pi$ from $s$ to $v$ in $G_s'$ that avoids $e$. This path does not exist in $G_s$, so $\pi$ contains $(x,y)$.
Consider the subpath $\pi_1$ of $\pi$ from $y$ to $v$. Path $\pi_1$ exists in $G_s$ and avoids $e$. So, Lemma \ref{lemma:partition-paths} implies that 
all vertices in $\pi_1$
($y$ included)
are descendants of $v$ in $D$, since otherwise $\pi_1$ would have to include $e$.
\end{proof}

\begin{corollary}
\label{corollary:cancelled-bridge-3}
A bridge $e=(u,v)$ of $G_s$ is canceled by the insertion of edge $(x,y)$ if and only if $ \mathit{depth}(\mathit{nca}(x,y)) \le \mathit{depth}(u)$ and there is a path $\pi$ in $G$ from $y$ to $v$ such that $\mathit{depth}(u) < \mathit{depth}(w)$ for all $w \in \pi$.
\end{corollary}

By Corollary \ref{corollary:cancelled-bridge-3}, we can use
the incremental algorithm of \cite{dyndom:2012} to
detect canceled bridges, without affecting the $O(mn)$ bound.
Indeed, suppose $e=(u,v)$ is a canceled bridge. By Lemma \ref{lemma:canceled-bridge},
$y$ is a descendant of $v$ in $D$ and in the same
$e$-dominated component as $v$. Hence, $v$ will be visited by the search from $y$.

If a bridge $(u,v)$ is locally canceled,
there can be vertices in $D_v$ that are not scanned, and that after the insertion will be located in $D_{u}$, without having
their depth changed.
This is a difficult case for our analysis: fortunately,
the following lemma shows that this case can happen only $O(n)$ times overall.

\begin{lemma}
\label{lemma:locally}
Suppose $(u,v)$ is a bridge of $G_s$ that is locally canceled by the insertion of edge $(x,y)$.
Then $(u,v)$ is no longer a strong bridge in $G$ after the insertion.
\end{lemma}
\begin{proof}
Since $v$ is visited but not affected, we have that $u =\mathit{nca}(x,y)$. Then $x$ is a descendant of $u$ in $D$, and $y \not\in D[u,x]$. Hence, there is a path $\pi_1$ from $u$ to $x$
that does not contain $y$.
Since $y$ is affected,
Lemma \ref{lemma:canceled-bridge}
implies that there is a path $\pi_2$ from $y$ to $v$ that contains only descendants of $v$.
So, $\pi_1\cdot (x,y)\cdot\pi_2$ is a path from $u$ to $v$ that avoids $(u,v)$ in $G$ after the insertion of $(x,y)$.
\end{proof}

Note that a canceled bridge that is not locally canceled may still appear as a bridge in $G_s^R$ after the insertion of edge $(x,y)$.
Next we provide some lemmata that help us to identify the necessary changes in the
auxiliary components of the affected subgraphs and $\widehat{G}_{r_z}$.

\begin{lemma}
\label{lemma:affected-subtrees}
Let $v$ be a vertex that is affected by the insertion of edge $(x,y)$. Then $r_v$ is on the path $D[r_z,r_y]$.
\end{lemma}
\begin{proof}
Vertices $x$ and $y$ are descendants of $z$, so $r_y$ and $r_x$ are descendants of $r_z$.
Since $v$ is affected, $v$ is a descendant of $z$, and by Lemma \ref{lemma:insert-affected} there is a path $\pi$ from $y$ to $v$ such that $\mathit{depth}(d(v)) < \mathit{depth}(w)$ for all $w \in \pi$.
Thus, $\pi$ does not contain the bridge $(d(r_v),r_v)$, so by Lemma \ref{lemma:partition-paths}, $y$ is a descendant of $r_v$. Then $r_v$ is a descendant of $r_z$ and an ancestor of $r_y$.
\end{proof}

\begin{lemma}\emph{(\cite{DomCert:TALG})}
\label{lemma:strongly-connected-siblings}
Let $S$ be the set of vertices of a strongly connected subgraph of $G$. Then $S$ consists of a set of siblings in $D$ and possibly some of their descendants in $D$.
\end{lemma}

In the following, we assume that
bridge $(p,q)$ is not locally canceled after the insertion of $(x,y)$ and that $z=\mathit{nca}(x,y)$.

\begin{lemma}
\label{lemma:affected-component}
Let $C$ be an affected auxiliary component of an auxiliary graph $\widehat{G}_r$ with $r \not= r_z$.
Then $C$ consists of a set of affected siblings in $D$ and possibly some of their affected or scanned descendants in $D$.
\end{lemma}
\begin{proof}
By Lemma \ref{lemma:strongly-connected-siblings}, $C$ consists of a set of siblings $S$ in $D$ and possibly some of their descendants in $D$.
Also, Lemma \ref{lemma:affected-reachable} implies that all vertices in $C$ are scanned.
So, it suffices to show that all siblings in $S$ are affected. Let $v$ be an affected vertex in $C$. Consider any sibling $u \in S$. Since $C$ is strongly connected,
there is a path $\pi_2$ from $v$ to $u$ containing only vertices in $C$.
Since $(p,q)$ is not locally canceled, $u \not= q$ or $z \not= p$, so all vertices $w$ on $\pi_2$ satisfy $\mathit{depth}(z)<\mathit{depth}(d(w))$.
Let $\pi_1$ a path from $y$ to $v$ that satisfies Lemma \ref{lemma:insert-affected}.
Then $\pi_1 \cdot \pi_2$ is a path from $y$ to $u$ that also satisfies Lemma \ref{lemma:insert-affected}. Hence $u$ is affected.
\end{proof}

An auxiliary component is \emph{scanned} if it contains a scanned vertex.
As with vertices,
affected auxiliary components are also scanned (the converse is not necessarily true).

\begin{lemma}
\label{lemma:scanned-component}
Let $C$ be a scanned auxiliary component of an auxiliary graph $\widehat{G}_r$ with $r \not= r_z$.
Then all vertices in $C$ are scanned.
\end{lemma}
\begin{proof}
The fact that $(p,q)$ is not locally canceled implies that $q \not\in C$ or $p \not= z$.
So, for each vertex $w$ in $C$ we have $\mathit{depth}(z)<\mathit{depth}(d(w))$,which implies that $w$ is scanned.
\end{proof}

We say that a vertex $v$ is \emph{moved} if it is located in an auxiliary graph $\widehat{G}_r$ with $r \not= r_z$ before the insertion of $(x,y)$,
and in $\widehat{G}_{r_z}$ after the insertion.
Lemmata \ref{lemma:affected-component} and \ref{lemma:scanned-component} imply that if an auxiliary component $C$ contains
a moved vertex, then all vertices in the component are also moved. 
We call such an auxiliary component \emph{moved}.
Now we describe how to find the moved auxiliary components that need to be merged.
Let $H$ be the subgraph of $G$ induced by the scanned vertices in $D(q)$. We refer to $H$ as the \emph{scanned subgraph}.

\begin{lemma}
\label{lemma:scanned-subgraph}
Let $\zeta$ and $\xi$ be two distinct roots in the bridge decomposition of $D$, such that $\zeta, \xi \not= r_z$,
and $D_{\zeta}$ and $D_{\xi}$ are contained in $D(q)$.
Let $C_{\zeta}$ and $C_{\xi}$ be scanned components in $\widehat{G}_{\zeta}$ and $\widehat{G}_{\xi}$, respectively.
Then $C_{\zeta}$ and $C_{\xi}$ are strongly connected in $G[D(q)]$ if and only if they are strongly connected in $H$.
\end{lemma}
\begin{proof}
Clearly, $C_{\zeta}$ and $C_{\xi}$ are strongly connected in $G[D(q)]$ if they are strongly connected in $H$, so it remains to prove the converse.
Suppose $C_{\zeta}$ and $C_{\xi}$ are strongly connected in $G[D(q)]$.
Then, there is a path $\pi$ in $G[D(q)]$ from a vertex in $C_{\xi}$ to
a vertex in $C_{\zeta}$.
The fact that $(p,q)$ is not locally canceled implies that $q \not\in C_{\zeta} \cup C_{\xi}$ or $z \not= p$.
Then, for each vertex $w$ on $\pi$ we have $\mathit{depth}(z)<\mathit{depth}(d(w))$, which implies that $w$ is scanned.
Hence $\pi$ exists in $H$.
Similarly, there is a path in $G[D(q)]$ from a vertex in $C_{\zeta}$ to
a vertex in $C_{\xi}$ that is also contained in $H$.
\end{proof}

Now we introduce a dummy root $r^{\ast}$ in $H$, together with an edge $(v,r^{\ast})$ for each scanned vertex $v$ that has a leaving edge $(v,w)$ such that $w \in D_z$ and $w$ is in the auxiliary component of $p$ in $\widehat{G}_{r_z}$.
We denote this graph by $H^{\ast}$.

\begin{lemma}
\label{lemma:scanned-subgraph-2}
A scanned vertex $v \not \in D_z$ is strongly connected in $G[D(r_z)]$ to a vertex $w \in D_z$  if and only if $r^{\ast}$ is reachable from $v$ in $H^{\ast}$.
In this case, $v$ and $p$ are also strongly connected in
$G[D(r_z)]$.
\end{lemma}
\begin{proof}
Let $v \not \in D_z$ be a scanned vertex. By Lemma \ref{lemma:affected-subtrees}, $v$ is in $D(q)$, hence a descendant of $p$ in $D$.
Suppose $r^{\ast}$ is reachable from $v$ in $H^{\ast}$. Then, there is a path $\pi_1$ in $G[D(r_z)]$ from $v$ to a vertex
$w \in D_z$, where $w$ is in the auxiliary component of $p$ in $\widehat{G}_{r_z}$. Since $p$ and $w$ are in the same auxiliary component,
there is a path $\pi_2$ from $w$ to $p$ in $G[D(r_z)]$. Also, since $v$ is a descendant of $p$ in $D$, there is a path $\pi_3$ from $p$ to $v$ in $G[D(r_z)]$.
Paths $\pi_1$ and $\pi_2\cdot \pi_3$ imply that $w$ and $v$ are strongly connected in $G[D(r_z)]$.

Conversely, let
$w$ be a vertex in $D_z$ that is strongly connected in $G[D(r_z)]$ to a scanned vertex $v \not \in D_z$.
Let $\pi_1$ be a path from $w$ to $v$ in $G[D(r_z)]$, and let $\pi_2$ be a path from $v$ to $w$ in $G[D(r_z)]$.
From Lemma \ref{lemma:affected-reachable} we have that $v \in D(q)$, hence
Lemma \ref{lemma:partition-paths} implies that $\pi_1$ contains $(p,q)$ so $p \in \pi_1$.
Thus $p$ and $w$ are also strongly connected in $G[D(r_z)]$.
Now let $w'$ be the first vertex on $\pi_2$ that is in $D_z$, and let $(t,w')$ be the edge in $\pi_2$ that enters $w'$.
Then $w'$ and $w$ are also  strongly connected in $G[D(r_z)]$.
Also, since $(p,q)$ is not locally canceled, we have $q \not\in \pi_2$ or $z \not= p$.
Then, for each vertex $t'$ on the part of $\pi_2$ from $v$ to $t$ we have $\mathit{depth}(z)<\mathit{depth}(d(t'))$, which implies that $t'$ is scanned.
Hence $\pi_2$ exists in $H$, so by construction, $v$ reaches $r^{\ast}$ in $H^{\ast}$.
\end{proof}

\subsection{The Algorithm}

We describe next our incremental algorithm for maintaining the $2$-edge-connected blocks of a strongly connected digraph $G$. We refer to this algorithm as \textsf{SCInc2ECB}$(G)$.
We initialize the algorithm and the associated data structures by executing the labeling algorithm of Section \ref{sec:labeling-algorithm}.
Algorithm \textsf{Initialize}$(G,s)$, shown below, computes the dominator tree $D$, the set of bridges $Br$ of flow graph $G_s$,
the bridge decomposition $\mathcal{D}$ of $D$, and the corresponding auxiliary graphs $\widehat{G}_r$.
Finally, for each auxiliary graph $\widehat{G}_r$,
it finds its auxiliary components, computes the labels $r_w$ and $c_w$ for each vertex $w \in V_r$, and
initializes an \textsf{IncSCC} data structure.
The execution of \textsf{Initialize}$(G^R,s)$ performs analogous steps in the reverse flow graph $G^R_s$.

\begin{algorithm}[h!]
\LinesNumbered
\DontPrintSemicolon
Set $s$ to be the designated start vertex of $G$.\;

Compute the dominator tree $D$ and the set of bridges $\mathit{Br}$ of the corresponding flow graph $G_s$. \;

Compute the bridge decomposition $\mathcal{D}$ of $D$.\;
\ForEach{root $r$ in $\mathcal{D}$}
{
    Compute the auxiliary graph $\widehat{G}_r$ of $r$.\;
	Compute the strongly connected components in $\widehat{G}_r$.\;
    \ForEach{strongly connected component $C$ in $\widehat{G}_r$}
    {
        Choose a vertex $v \in C$ as the canonical vertex of the auxiliary component $C$.\;
       \ForEach{vertex $w \in C$}
       {
           Set $r_w = r$ and $c_w = v$.
       }
    }
    Initialize a \textsf{IncSCC} data structure for $\widehat{G}_r$.\;
 }
 \caption{\textsf{Initialize}$(G,s)$}
\end{algorithm}

\begin{algorithm}
\LinesNumbered
\DontPrintSemicolon

Let $s$ be the designated start vertex of $G$, and let $e=(x,y)$.\;

Compute the nearest common ancestor $z$ and $z^R$ of $x$ and $y$ in $D$ and $D^R$ respectively.\;

Update the dominator trees $D$ and $D^R$, and return the lists $S$ and $S^R$ of the vertices that were scanned in $D$ and $D^R$ respectively.\;

	\eIf{a bridge is locally canceled in $G_s$ or in $G_s^R$}
    {
	  Execute \textsf{Initialize}$(G,s)$ and \textsf{Initialize}$(G^R,s)$.
    }
	{					
      Execute \textsf{UpdateAC}$(\mathcal{D},z,x,y,S)$ and \textsf{UpdateAC}$(\mathcal{D}^R,z^R, y,x,S^R)$.
    }

\caption{\textsf{SCInsertEdge}$(G,e)$}
\end{algorithm}

\begin{algorithm}[h!]
\LinesNumbered
\DontPrintSemicolon

Let $r_z$ be root of the tree $D_z$ in $\mathcal{D}$ that contains $z$.\;

Let $c_{x'}$ be the canonical vertex of the nearest ancestor $x'$ of $x$ in $D$ such that $x' \in D_z$.\;

Let $(p,q)$ be the first bridge on the path $D[z,y]$, and let $c_p$ be the canonical vertex of $p$.\;

Form the scanned graph $H^{\ast}$ that contains the scanned vertices $S \setminus D_z$ and the edges among them.\;

Compute the strongly connected components $\mathcal{C}$ of $H^{\ast} \setminus r^{\ast}$ and order them topologically.\;

Compute the components $\mathcal{C}^{\ast}$ of $\mathcal{C}$ that reach $r^{\ast}$ in $H^{\ast}$.\;

\ForEach{strongly connected component $C$ in $\mathcal{C}^{\ast}$ that is moved}
{
Merge $C$ with the component of $c_p$.\;
}

\ForAll{strongly connected components in $\mathcal{C} \setminus \mathcal{C}^{\ast}$ that are moved}
{
Insert the components in the topological order of $\widehat{G}_{r_z}$ just after the component of $c_p$.\;
}

\ForEach{vertex $w \in S$}
{
           \lIf{$w$ is moved to $\widehat{G}_{r_z}$} { set $r_w = r_z$.}
}

Update the lists of out edges in the \textsf{IncSCC} data structures of $\widehat{G}_{r_z}$ and of the affected auxiliary graphs.\;

Insert edge $(c_{x'}, y)$ in the list of outgoing edges of $c_{x'}$ and update the \textsf{IncSCC} data structure of $\widehat{G}_{r_z}$.\;

\caption{\textsf{UpdateAC}$(\mathcal{D},z,x,y,L)$}
\end{algorithm}

When a new edge $e=(x,y)$ is inserted, algorithm \textsf{SCInc2ECB} executes procedure $\mathsf{SCInsertEdge}(G,e)$, which updates dominator trees $D$ and $D^R$, together with the corresponding bridge decompositions.
It also finds the set of scanned vertices in $G_s$ and $G_s^R$. If a bridge of  $D$ or $D^R$ is locally cancelled, then we restart the algorithm by executing
\textsf{Initialize}. Otherwise, we need to update the auxiliary components in $G_s$ and $G_s^R$. These updates are handled by procedure
\textsf{UpdateAC}. Before describing \textsf{UpdateAC}, we provide some details on the implementation of  the \textsf{IncSCC} data structures, which
maintain the auxiliary components of each auxiliary graph $\widehat{G}_r$ using
the ``one-way search'' structure of \cite[Sections 2 and 6]{Haeupler:IncTopOrder:TALG}.
Since we need to insert and delete canonical vertices, we augment this data structure as follows.
We maintain the canonical vertices of each auxiliary component in a linked list $L_r$, arranged according to the given topological order of $\widehat{G}_r$.
For each vertex $v$ in $L_r$, we also maintain a rank in $L_r$ which is an integer in $[1,n]$ such that for any two canonical vertices $u$ and $v$ in $L_r$, $\mathit{rank}(u) < \mathit{rank}(v)$ if and only if $u$ precedes $v$ in $L_r$.
The ranks of all vertices can be stored in a single array of size $n$.
Also, with each canonical vertex $w$, we store a pointer to the location of $w$ in $L$.
We represent $L_r$ with a doubly linked list so that we can insert and delete a canonical vertex in constant time.
When we remove vertices from a list $L_r$ we do not need to update the ranks of the remaining vertices in $L_r$.
The insertion of an edge $(x,y)$ may remove vertices from various lists $L_r$, but may insert vertices only in $L_{r_z}$.
After these insertions, we recompute the ranks of all vertices in $L_{r_z}$ just by traversing the list and assigning rank $i$ to the $i$-th vertex in the list.
We maintain links between an original edge $e$, stored in the adjacency lists of $G$, and at most one copy of $e$ in a $\mathit{out}$ list of \textsf{IncSCC}.
This enables us to keep for each shortcut edge $e'=(v',w)$ a one-to-one correspondence with
the original edge $e=(v,w)$ that created $e'$.
We do that because if an ancestor of $v$ is moved to the auxiliary graph $\widehat{G}_{r_z}$ that contains $v'$ ($v'=p$ in Figure \ref{figure:ShortcutEdge}),
then $e$ may correspond to a different shortcut edge or it may even
become an ordinary edge of $\widehat{G}_{r_z}$. Using this mapping we can update the $\mathit{out}$ lists of \textsf{IncSCC}.
To initialize the \textsf{IncSCC} structure of an auxiliary graph, we compute a topological order of the auxiliary components in $\widehat{G}_r$, and create the
list of outgoing edges $\mathit{out}(v)$ for each canonical vertex $v$.

If inserting edge $(x,y)$ does not locally cancel a bridge in $G_s$ and $G^R_s$, then we
update the auxiliary components of $G_s$ using procedure
\textsf{UpdateAC}$(\mathcal{D},z,x,y,S)$, where $\mathcal{D}$ is the updated bridge decomposition of $D$,
$z=\mathit{nca}(x,y)$,
and $S$ is a list of the vertices scanned during the update of $D$.
We
do the same to update the auxiliary components of $G_s^R$.
Procedure \textsf{UpdateAC} first computes the auxiliary components that are moved to $\widehat{G}_{r_z}$,
possibly merging some of them, and then inserts the edge $(x,y)$ as an original or a shortcut edge
of $\widehat{G}_{r_z}$, depending on whether $x \in D_{r_z}$ or not. Note that the insertion of $(x,y)$ may cause
the creation of a new auxiliary component in $\widehat{G}_{r_z}$.
Now we specify some further details in the implementation of \textsf{UpdateAC}. The vertices that are moved to $\widehat{G}_{r_z}$ are
the scanned vertices in $S$ that are not descendants of a strong bridge.
Hence, we can mark the vertices that are moved to $\widehat{G}_{r_z}$ during the search for affected vertices.
The next task is to update the $\mathit{out}$ lists of the canonical vertices in $\widehat{G}_{r_z}$ and the affected auxiliary graphs.
We process the list of scanned vertices $S$ as follows.
Let $v$ be such a vertex. If $v$ is not marked, i.e., is not moved to $\widehat{G}_{r_z}$, then we process the edges leaving $v$; otherwise, we process
both the edges leaving $v$ and the edges entering $v$.
Suppose $v$ is marked. Let $(v,w)$ be an edge leaving $v$ in $G$. If $w$ is also in $\widehat{G}_{r_z}$ after the insertion, then
we add the edge $(v,w)$ in $\mathit{out}(f(v))$. Moreover, if $w$ is not in $S$, then it was already located in $\widehat{G}_{r_z}$ before the insertion,
so we delete the shortcut edge stored in $\mathit{out}(f(p))$.
If $w$ is not in $\widehat{G}_{r_z}$ after the insertion, then $(v,w)$ is a bridge in $D$ and we do nothing.
Now consider an edge $(w,v)$ entering $v$ in $G$. If $w$ is scanned, then we will process $(w,v)$ while processing the edges leaving $w$.
Otherwise, $w$ remains a descendant of $p$, so we insert the edge $(w,v)$ in $\mathit{out}(f(p))$.
Now we consider the unmarked scanned vertices $v$. Let $(v,w)$ an edge leaving $v$ in $G$.
If $w \in D_z$, we insert the edge $(v,w)$ into $\mathit{out}(f(v'))$, where $v'$ is the nearest marked ancestor of $v$ in $D$.
Otherwise, if $w \notin D(r_z)$, the edge $(v'',w)$, where $v''$ is the nearest ancestor of $v$ in $D_w$, already exists since $v$ was a descendant of $v''$ before the insertion of $(x,y)$.
Next, we consider the updates in the $L_r$ lists and the vertex ranks. While we process $S$, if we encounter a moved canonical vertex $v \in S$
that was located in
an auxiliary graph $\widehat{G}_{r}$ with $r_z \not= r$, then we delete $v$ from $L_r$.
Note that we do not need to update the ranks of the remaining vertices in lists $L_r$ with $r \not= r_z$.
To update $L_{r_z}$, we insert the moved canonical vertices of the SCCs
in $\mathcal{C} \setminus \mathcal{C}^{\ast}$, in a topological order of
$H = H^{\ast} \setminus r^{\ast}$, just after $f(p)$.
Then we traverse $L_{r_z}$ and update the ranks of the canonical vertices.
The final step is to actually insert edge $(x,y)$ in the \textsf{IncSCC} data structure of  $\widehat{G}_{r_z}$.
We do that by adding $(x,y)$ in $\mathit{out}(f(x'))$, where $x'$ is the nearest ancestor of $x$ in $D_{z}$.
If $\mathit{rank}(f(x')) > \mathit{rank}(f(y))$, then we execute the forward-search procedure of \textsf{IncSCC}.

\begin{figure}[t]
\begin{center}
\includegraphics[trim={0 0 0 7cm}, clip=true, width=1.0\textwidth]{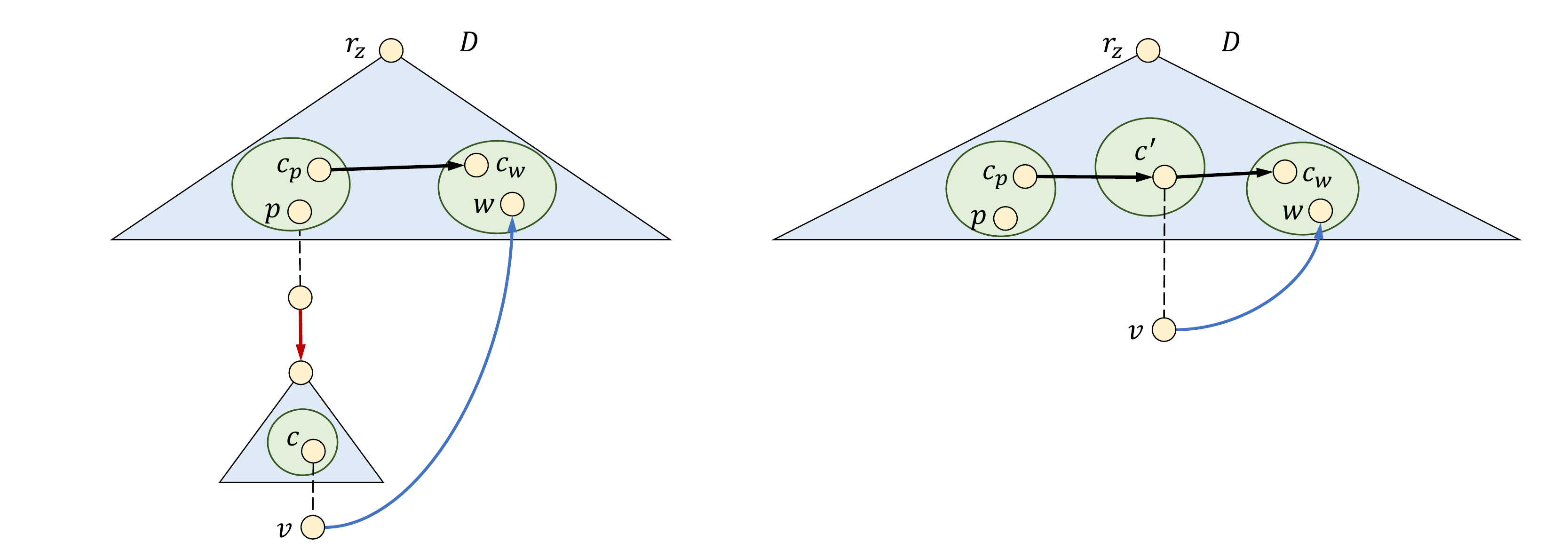}	
\caption{Before the insertion of $(x,y)$, edge $(v,w)$ corresponds to the shortcut edge $(p,w)$ of $\widehat{G}_{r_z}$, and is stored in $\mathit{out}(c_p)$.
An auxiliary component with canonical vertex $c$ is affected by the insertion of $(x,y)$ and is merged into a component with canonical vertex $c'$
($c'=c$ if the component is moved without merging with another component).
Now $c'$ becomes the canonical vertex of the nearest ancestor of $v$ in $D_z$, and edge $(v,w)$ is stored as a shortcut edge in $\mathit{out}(c')$.} 
\label{figure:ShortcutEdge}
\end{center}
\end{figure}

\begin{lemma}
\label{lemma:update-correctness}
Algorithm \textsf{SCInc2ECB} is correct.
\end{lemma}
\begin{proof}
It suffices to show that \textsf{UpdateAC} correctly maintains: (i) the auxiliary components, and (ii) the \textsf{IncSCC} structure of each auxiliary graph.
The fact that (i) holds follows from Lemmata \ref{lemma:scanned-subgraph} and \ref{lemma:scanned-subgraph-2}.

To prove (ii), we need to show that the topological order and the $\mathit{out}$ lists of each auxiliary graph are updated correctly in
lines 10--12 and 16 of \textsf{UpdateAC}.

Suppose $\widehat{G}_{r}$ is an auxiliary graph that is neither affected nor contains $z$. We argue that no shortcut edge
in $\widehat{G}_{r}$ needs to be replaced.
By Lemma \ref{lemma:affected-subtrees},
we only need to consider edges $(u,v)$ such that $v \in D_r$ and $u$ is a descendant of an affected vertex.
We distinguish three cases for $r$. If $r$ is not an ancestor or a descendant of $r_z$,
then by Lemma \ref{lemma:partition-paths} no such edge $(u,v)$ exists.
Similarly, if $r$ is a descendant of $r_z$ but not on $D[r_z,r_y]$, then by Lemmata  \ref{lemma:partition-paths} and \ref{lemma:affected-subtrees},
no such edge $(u,v)$ exists.
Finally, suppose that $r$ is an ancestor of $r_z$, and let $u'$ be the nearest ancestor of $u$ in $D_r$.
Then, $u'$ remains an ancestor of $r_z$ after the insertion, so the shortcut edge in  $\widehat{G}_{r}$
that corresponds to $(u,v)$ does not change.
In all three cases the $\mathit{out}$ lists of the \textsf{IncSCC} structure of $\widehat{G}_{r}$ remain valid,
and hence so does the topological order of its auxiliary components.

It remains to consider the affected
auxiliary graphs and $\widehat{G}_{r_z}$.
Let $\widehat{G}_{r}$ be an affected auxiliary graph with $r \not= r_z$.
Then, (i) holds for $\widehat{G}_{r}$ by
Lemmata \ref{lemma:affected-component} and \ref{lemma:scanned-component}.
Also, since in $\widehat{G}_{r}$ we only remove vertices, the topological order for the
components left in $\widehat{G}_{r}$ remains valid.
Furthermore, for each scanned canonical vertex $w$, we delete from $\mathit{out}(f(v))$ any
edge $(v,w)$ if $f(v)$ is a canonical vertex in $\widehat{G}_{r}$ that is not scanned.
This implies that (ii) also holds for $\widehat{G}_{r}$.

Finally, consider $\widehat{G}_{r_z}$.
From Lemma \ref{lemma:scanned-subgraph-2},
we have that lines 6--9 of \textsf{UpdateAC} correctly identify and update the auxiliary components
that are merged with components in $\widehat{G}_{r_z}$.
Also, by Lemma \ref{lemma:scanned-subgraph},
the remaining auxiliary components that are moved to $\widehat{G}_{r_z}$
are correctly identified in line 5.
So, (i) holds for $\widehat{G}_{r_z}$.
To prove that (ii) holds as well, consider a scanned canonical vertex $v$ that is moved to $\widehat{G}_{r_z}$.
By Lemma \ref{lemma:affected-subtrees}, $v$ is a descendant of $p$, hence
$\widehat{G}_{r_z}$ contains a path from $f(p)$ to $v$.
Let $w$ be a vertex that was in  $\widehat{G}_{r_z}$ before the insertion of edge $(x,y)$.
Then, Lemma \ref{lemma:partition-paths} implies that if there is a path in $G[D(q)]$ from $w$  to $v$,
then this path contains $p$.
Since $v$ remains a canonical vertex,
$v \not= f(p)$ and there is no path from $v$ to $f(p)$.
Therefore, the moved auxiliary components are ordered correctly in lines 10--12 of \textsf{UpdateAC},
and (ii) follows.
\end{proof}

\subsection{Running time of \textsf{SCInc2ECB}}
\label{sec:analysis}

We analyze the running time of Algorithm \textsf{SCInc2ECB}.
Recall that $G$ is a strongly connected digraph with $n$ vertices
that undergoes a sequence of edge insertions. We let $m$ be the total number of edges in $G$ after all insertions ($m \ge n$). First, we bound the time spent by \textsf{Initialize}. This procedure is called twice in the beginning of the \textsf{SCInc2ECB},
and twice after each time a bridge in $G_s$ or in $G_s^R$ is locally canceled.
Then, Lemma \ref{lemma:locally} implies that such an event can happen at most  $2(n-1)$ times. Hence, there are at most
$4n$ calls to \textsf{Initialize}, and since each execution takes $O(m)$ time,
the total time spent on \textsf{Initialize} is $O(mn)$.
Similarly, the dominator trees of $G_s$ and $G_s^R$ can be updated in total $O(mn)$ time~\cite{dyndom:2012}.
We next bound the total time required to update the auxiliary components.
Consider an execution of \textsf{UpdateAC}.
Let $\nu$ and $\mu$, respectively, be the number of scanned vertices, after the insertion of edge $(x,y)$, and their adjacent edges.
The time to compute the affected subgraph $H^{\ast}$, compute the SCCs
of $H^{\ast} \setminus r^{\ast}$, and
the vertices that reach $r^{\ast}$ is $O(\nu +\mu)$.
In the same time, we can update the auxiliary components of $\widehat{G}_{r_z}$ and of the affected auxiliary graphs, their corresponding
topological orders, and the $\mathit{out}$ lists of the corresponding \textsf{IncSCC} data structures.
Since each scanned vertex $w$ is a descendant of an affected vertex, the depth of $w$ decreases by at most one. Hence,
the total time spent by \textsf{UpdateAC} for all insertions, excluding the execution of line 17, is $O(mn)$.
It remains to bound the time required by the \textsf{IncSCC} data structures to handle the edge insertions in line 17 of \textsf{UpdateAC}.
To do this, we extend the analysis from \cite{Haeupler:IncTopOrder:TALG}. We say that a vertex $v$ and an edge $e$ are \emph{related} if there is
a path that contains both $v$ and $e$ (in any order). Then, there are $O(mn)$ pairs of vertices and edges that can be related
in all \textsf{IncSCC} structures for every auxiliary graph.
We argue that each time the \textsf{IncSCC} structure traverses an edge (after the insertion in line 17 of \textsf{UpdateAC}),
the cost of this action can be charged to a newly-related vertex-edge pair.
Note that we cannot immediately apply the analysis in \cite{Haeupler:IncTopOrder:TALG}, since here we have the complication that
vertices and edges can be inserted to and removed from the \textsf{IncSCC} structures.
Consider a vertex $w$ and an edge $e=(u,w)$.
Call the pair $\langle v, e \rangle$ \emph{active} if $v$ and $e$ are in the same auxiliary graph $\widehat{G}_r$, and \emph{inactive} otherwise.
Note that since we identify
shortcut edges with their corresponding original edge, $e$ may actually appear in $\widehat{G}_r$ as an edge $(u',w)$, where $u'$ is the nearest ancestor
of $u$ in $D_r$. This fact, however, does not affect our analysis.

\begin{lemma}
\label{lemma:related-pairs}
The total number of edge traversals made during the forward searches in all \textsf{IncSCC} data structures is $O(mn)$.
\end{lemma}
\begin{proof}
To prove the bound, it suffices to show that
in all \textsf{IncSCC} data structures the total number of unrelated
$\langle v, e \rangle$ pairs that are ever created is $O(mn)$. Consider an active pair $\langle v, e \rangle$ that becomes related in
$\widehat{G}_r$. Then there is some path $\pi$ in $G[D(r)]$ that contains both $v$ and $e$. Suppose that the pair $\langle v, e \rangle$
later becomes active but unrelated in an auxiliary graph $\widehat{G}_{r'}$, where $r'$ may be vertex $r$.
Then $\pi$ does not exist in $G[D(r')]$, which implies that some vertices of $\pi$ are not descendants of $r'$.
Then, by Lemma \ref{lemma:partition-paths}, $\pi$ must contain the bridge $(d(r'),r')$. Since $\pi$ exists in $G[D(r)]$, the bridge $(d(r'),r')$ was a descendant of $r$ before some insertion,
and then became an ancestor of $v$. But this is impossible, since after an edge insertion, the new parent $d'(v)$ of $v$
is on the path $D[s,d(v)]$.
Hence, once a $\langle v, e \rangle$ pair becomes related, it can never become unrelated. The bound follows.
\end{proof}

\begin{lemma}
\label{lemma:IncSCC-updates}
The total time to update all the \textsf{IncSCC} data structures is $O(mn)$.
\end{lemma}
\begin{proof}
Updating the lists of out edges in the \textsf{IncSCC} data structures,
and inserting or deleting canonical vertices can be charged to the cost of
updating the dominator tree, and is thus $O(mn)$.
By Lemma \ref{lemma:related-pairs}, all edge insertions that do not
trigger merges of auxiliary components can be handled in $O(mn)$ time.
The number of edge insertions that trigger merges of auxiliary
components is at most $n-1$, and each such insertion can be handled
in $O(m+n)$ time, excluding unite operations. Taking into account
also the total time for all unite operations yields the lemma.
\end{proof}

\begin{theorem}
\label{theorem:SCInc2ECB-time}
The total running time of Algorithm \textsf{SCInc2ECB} for a sequence of edge insertions in a strongly connected digraph with $n$ vertices
is $O(mn)$, where $m$ is the total number of edges in $G$ after all insertions.
\end{theorem}

\section{Incremental algorithm for $2$-edge-connected blocks in general graphs}
\label{sec:incremental-general}

In this section, we show how to extend our 
algorithm to general digraphs that are not necessarily strongly connected.
Let $G$ be input digraph that undergoes edge insertions.
We will design a data structure that maintains the $2$-edge-connected blocks of $G$ and can report if any two query vertices are $2$-edge-connected.
We use a two-level data structure. The top level maintains the strongly connected components of $G$ with the use of a \textsf{IncSCC} data structure of Section \ref{sec:incremental-scc}.
We refer to this data structure as \textsf{TopIncSCC}

If the insertion of an edge creates a new component $C$, algorithm \textsf{TopIncSCC} finds the vertices in the new component and updates the condensation of $G$.
Let $C_1, C_2, \ldots, C_j$ be the components that were merged into $C$ after the insertion of an edge. We choose the canonical (start) vertex of $C$ to be the start vertex of the
largest component $C_i$. We refer to this component $C_i$ as the \emph{principal component of $C$}.

\begin{algorithm}
\LinesNumbered
\DontPrintSemicolon

Let $e=(x,y)$.\;
\eIf{$x$ and $y$ are in the same component $C$}
    {
	  Execute \textsf{SCInsertEdge}$(G[C],e)$.
    }
	{					
	  Insert $e$ into \textsf{TopIncSCC}.\;
      \If{a new component $C$ is created}
      {
        Let $s$ be the designated start vertex of the largest component merged into $C$.\;
        Execute \textsf{Initialize}$(G[C],s)$ and \textsf{Initialize}$(G^R[C],s)$.
      }
    }
\caption{\textsf{InsertEdge}$(G,e)$}
\end{algorithm}

Now we bound the running time of \textsf{Inc2ECB}, excluding the time required by \textsf{InsertEdge}.
The total time required to maintain the \textsf{TopIncSCC} structure is $O(mn)$.
The total number of new components that can be found by \textsf{TopIncSCC}
is at most $n-1$. When such an event occurs, algorithm \textsf{Inc2ECB} makes two calls to \textsf{Initialize}, and
each such call takes $O(m)$ time to initialize the bottom-level structure, i.e., the \textsf{SCInc2ECB} data structure for the new component.
Hence, the total time required by \textsf{Inc2ECB}, excluding the calls to \textsf{InsertEdge}, is $O(mn)$.

Next we bound the time spent on calls to \textsf{InsertEdge}.
First, we bound the total time required to maintain all dominator trees, for each component created by the main algorithm \textsf{Inc2ECB}.
Recall from Section \ref{sec:affected} that a vertex $v$ is \emph{scanned} if it is a descendant of an affected vertex.
Each scanned vertex $v$ incurs a cost of $O(\mathrm{degree}(v))$, thus we need to bound the number of times a vertex can be scanned.
Each time a vertex is scanned, its depth in the dominator tree decreases by at least one.
Let $C$ be the current component containing $v$, and let $C'$ be a new component that $C$ is merged into following an edge insertion.
If $C$ is the principal subcomponent of $C'$ then the depth of $v$ may only decrease. Otherwise, the depth of $v$ may increase.

We define the \emph{effective depth} of $v$ 
after merging $C$ into $C'$ to be zero, if $C$ is the principal subcomponent of $C'$, and equal to the depth of $v$ in the dominator tree of $G_s[C']$ otherwise.
To bound the total amount of work needed to maintain the dominator trees of all components, we compute the sum of the effective depths of $v$ in all the components that $v$ is contained throughout the execution of algorithm \textsf{Inc2ECB}. We refer to this sum as the \emph{total effective depth of $v$}, denoted by $\mathit{ted}(v)$.

\begin{lemma}
\label{lemma:effective-depth}
The total effective depth of any vertex $v$ is $O(n)$.
\end{lemma}
\begin{proof}
Suppose that the component of $v$ was merged $k$ times as a non-principal component. Let $n_i$ be the number of vertices in the $i$th
component that contains $v$ that was later merged as a non-principal component. The effective depth of $v$ in this component is less than $n_i$.
When a non-principal component is merged, the resulting component has at least $2|C|$ vertices.
Thus, the total effective depth of $v$ is $\mathit{ted}(v) \le \sum{n_i}$, where $n_i \le n$ and $n_{i+1} \ge 2 n_{i}$.
To maximize the sum, set $n_k=n$ and $n_{i} = n_{i+1}/2$, so we have $\mathit{ted}(v) \le n + n/2 + \dots + 1 = 2n$.
\end{proof}

By Lemma \ref{lemma:effective-depth}, each vertex $v$ incurs a total cost of $O(n\mathrm{degree}(v))$ while maintaining the dominator trees for all the components
of the \textsf{TopIncSCC} structure. Hence, the time spent on updating these dominator trees is $O(mn)$.
Next, we analyze the total time spent on \textsf{Initialize} through calls made by \textsf{InsertEdge}.
Algorithm \textsf{InsertEdge} calls \textsf{Initialize} twice for each bridge that is locally canceled, hence at most twice the total number of strong bridges that
appear throughout the execution of \textsf{Inc2ECB}.

\begin{lemma}
\label{lemma:strong-bridges-count}
Throughout the execution of Algorithm \textsf{Inc2ECB} at most $2(n-1)$ strong bridges can appear.
\end{lemma}
\begin{proof}
When Algorithm \textsf{Inc2ECB} merges $k$ strongly connected components $C_1, C_2, \dots, C_k$ into a new strongly connected component $C$, then the new strong bridges that appear in $G[C]$ connect two different components $C_i$. This is because an edge $(u,v)$ of a subgraph $G[C_i]$ cannot be a strong bridge in $G[C]$ if it was not a strong bridge in $G[C_i]$.

Let $H$ be the multigraph that results from $G[C]$ after contracting each component $C_i$ into a single vertex. We claim that each new strong bridge of $(u,v)$ of $G[C]$ corresponds to a strong bridge $(C_i,C_j)$ in $H$, where $u \in C_i$ and $v \in C_j$. Multigraph $H$ contains $(C_i,C_j)$ by construction, and there is a unique edge $(C_i,C_j)$ in $H$. Indeed, if there was another edge $(u',v')$ in $G[C]$ with  $u' \in C_i$ and $v' \in C_j$, then $(u,v)$ could not be a strong bridge since $G[C]$ would have a path from $u$ to $v$, formed by a path from $u$ to $u'$ in $G[C_i]$, edge $(u',v')$, and a path from $v'$ to $v$ in $G[C_j]$, which avoids $(u,v)$.
If $(C_i,C_j)$ is not a strong bridge in $H$, then there is a path $P$ in $H$ from $C_i$ to $C_j$ that avoids $(C_i,C_j)$.
Hence, $G[C]$ has at most $2(k-1)$ new strong bridges.

Thus, we charge at most two new strong bridges each time a component is merged into a larger component. Since there are at most $n-1$ such merges, the total number of  strong bridges that can appear during the sequence of insertions is at most $2(n-1)$.
\end{proof}

Hence, by Lemma \ref{lemma:strong-bridges-count}, the total time spend on calls \textsf{Initialize} via \textsf{InsertEdge} is $O(mn)$.
Finally, we need to consider the time required to maintain the bottom \textsf{BottomIncSCC} structures. Unfortunately, it is no longer true that
an active vertex-edge pair that becomes related in a \textsf{BottomIncSCC} structure, remains related throughout the execution of the algorithm.
However, we can bound the number of times such a pair can change status from active and unrelated to active and related.

\begin{lemma}
\label{lemma:unrelated}
Let $v$ be a vertex, and let $e$ be an edge. The pair $\langle v, e \rangle$ can change status from active and unrelated to active and related at most $\log{n}$ times.
\end{lemma}
\begin{proof}
Suppose that $\langle v, e \rangle$ becomes active in the \textsf{BottomIncSCC} structure of some strongly connected component $G[C]$ of the top structure.
From Lemma \ref{lemma:related-pairs}, we have that in order for $\langle v, e \rangle$ to become active but unrelated, $C$ must be merged to another component as
a non-principal component. This can happen at most $\log{n}$ times, so the bound follows.
\end{proof}

By plugging the above bound in the proof of Lemma \ref{lemma:IncSCC-updates}, we get a total bound of $O(mn \log{n})$ for maintaining all \textsf{BottomIncSCC} structures.
We can improve this by employing a more advanced \textsf{BottomIncSCC} structure.  Namely, we can use the two-way search algorithm of Bender et al.~\cite{Bender:IncCycleDetection:TALG}.
The algorithm maintains for each canonical vertex $v$ a level $k(v)$. The levels are in a pseudo-topological order, i.e.,
if $(u,v)$ is an edge (original or formed by some contractions), then $k(f(u)) \le k(f(v))$.
The condensation of $G$ is maintained by storing for each canonical vertex $v$ a list $\mathit{out}(v)$
of the edges $(u,w)$ such that $f(u)=v$, to facilitate forward searches, and also a list $\mathit{in}(v)$ containing vertices $w$
such that $(f(w),v)$ is a loop or an edge of the current condensation with $k(f(w))=k(v)$, to facilitate backward searches.
It sets a parameter $\Delta = \min \{ m^{1/2}, n^{2/3} \}$ in order to bound the time spent during a backward search.
To insert an edge $(x,y)$, the algorithm computes $u=f(x)$ and $w=f(y)$. If $u=w$ or $k(u)<k(w)$ then the algorithm terminates.
Otherwise, it performs a backward search from $u$, visiting only canonical vertices at the same level as $u$.
During this search, loops or duplicate edges are not traversed. The backward search ends as soon as it
traverses $\Delta$ edges or runs out of edges traverse.
If the backward search traverses fewer than $\Delta$ edges and $k(w)=k(u)$ then the forward search is not executed.
Otherwise, if the backward search traverses $\Delta$ edges, or it traverses fewer than $\Delta$ edges but $k(w)<k(u)$,
then the algorithm executes a forward search from $w$. The forward search visits only vertices whose level increases.
Finally, if a cycle is detect during the backward or the forward search, the new component is formed.

We use the above algorithm to implement the \textsf{BottomIncSCC} data structures. As before, we augment
the $\mathit{out}$ and $\mathit{in}$ lists so that they also store original edges that become shortcut edges.
In the following, let $z=nca(x,y)$.
All occurrences of an original edge $(u,v)$, in the list of original edges leaving $u$, the list of original edges entering $v$,
and possibly in $\mathit{out}(f(u'))$ and $\mathit{in}(f(v))$, where $u'$ is the nearest ancestor of $u$
in $D_z$, are linked so that we can locate the shortcut edges in constant time.
To initialize a \textsf{BottomIncSCC} structure for an auxiliary graph $\widehat{G}_r$ (line 13 in procedure \textsf{Initialize}),
we compute the auxiliary components of $\widehat{G}_r$, and set the level of each canonical vertex $v$ to be one.
Then we create the $\mathit{out}$ and $\mathit{in}$ lists for the condensation of $\widehat{G_r}$.
In procedure \textsf{UpdateAC}, we do not change the levels of the vertices that are not moved. For any vertex that is moved to
$\widehat{G}_{r_z}$, we set its level to be equal to $k(f(p))$. Then we update the $\mathit{out}$ and $\mathit{in}$ lists (line 16),
as in Section \ref{sec:incremental-strongly-connected}. Finally, the new edge is added in $\widehat{G}_{r_z}$ (line 17), and we
execute the two-way search algorithm of Bender et al.
We refer to the implementation of algorithms \textsf{SCInc2ECB} and \textsf{Inc2ECB}, using the above \textsf{BottomIncSCC} data structure, as \textsf{SCInc2ECB-B} and \textsf{Inc2ECB-B}, respectively.

\begin{lemma}
\label{lemma:inc2ecb}
Algorithms \textsf{SCInc2ECB-B} and \textsf{Inc2ECB-B} are correct.
\end{lemma}
\begin{proof}
The only parts of the algorithms that are affected are the subroutines \textsf{Initialize} and \textsf{UpdateAC}.
In the former, we set the level of each canonical vertex $v$ to be $k(v)=1$. This is a valid initialization
since the levels are in pseudo-topological order. Moreover, lists $\mathit{out}$ and $\mathit{in}$
contain all the edges in the condensation of an auxiliary graph.
Now consider \textsf{UpdateAC}. Removing an auxiliary component and its canonical vertex from an auxiliary graph
and the \textsf{BottomIncSCC} structure, respectively, does not affect the fact that the levels are in
pseudo-topological order. Consider now the  updates in the \textsf{BottomIncSCC} structure of $\widehat{G}_{r_z}$.
In lines 7--9, we merge the auxiliary component of $f(p)$ with some components from other auxiliary graphs.
When we do that, we maintain $f(p)$ as the canonical vertex of the formed component, so its level does not change.
This means that the levels remain in pseudo-topological order for all canonical vertices that are already in $\widehat{G}_{r_z}$
before the insertion of $(x,y)$.
Consider now the insertion of canonical vertices in the \textsf{BottomIncSCC} structure
in lines 10--12. The level of all these vertices is set equal to $k(f(p))$.
Let $v$ be such a canonical vertex. Let $(u,v)$ be an edge entering $v$ from another canonical vertex of  $\widehat{G}_{r_z}$.
Then $u$ is either $f(p)$ or a vertex that was moved to $\widehat{G}_{r_z}$ together with $v$.
In both cases $k(u)=k(v)$. Now let $(v,w)$ be an edge out of $v$ entering another canonical vertex of  $\widehat{G}_{r_z}$.
Before the update,  $\widehat{G}_{r_z}$ contained a path from $f(p)$ to $w$, hence $k(f(p)) \le k(w)$.
Thus, the levels remain in pseudo-topological order.
\end{proof}

To prove the desired $O(mn)$ bound, we extend the analysis of \cite{Bender:IncCycleDetection:TALG}.
The analysis requires some additional definitions.
An original edge $(u, v)$ is \emph{live} if $u$ and $v$ are in
different components and \emph{dead} otherwise. A newly inserted edge that forms a
new component is dead.
The level of an edge $(u, v)$ is $k(f(u))$ if the edge is live, or equal to its highest level when it was live
if the edge is dead. If $(u,v)$ was never live, then it has no
level.
A component
is live if it corresponds to a vertex of the current condensation and dead otherwise.
A live component has level equal to the level of its canonical vertex. The level of a dead
component is its highest level when it was live. A vertex $w$ and a component $C$ are related
if there is a path that contains $w$ and a vertex in $C$.
Also, the number of components, live and dead, is at most $2n - 1$.

\begin{lemma}
\label{lemma:inc2ecb}
Algorithms \textsf{SCInc2ECB-B} and \textsf{Inc2ECB-B} run in $O(mn)$ time.
\end{lemma}
\begin{proof}
It suffices to show that the total time spent by the \textsf{BottomIncSCC} data structures is $O(mn)$.
The initialization of all such structures takes $O(m)$ time. Since, by Lemma \ref{lemma:strong-bridges-count},
the initialization occurs $O(n)$ times,
the total time is $O(mn)$.
It remains to bound the total insertion time.
We show that the following invariant, used in the analysis of \cite[Lemma 4.2]{Bender:IncCycleDetection:TALG},
is maintained: For any level $k>1$ and any level $j<k$,
any canonical vertex of level $k$ is related to at least $\Delta$ edges of level $j$
and at least $\sqrt{\Delta}$ components of level $j$.
The invariant is true after initialization, since all vertices, edges, and components have level at most one.
Bender et al. showed that the invariant is maintained after the insertion of an edge (line 17 of \textsf{UpdateAC}), so it remains to
show that the invariant is also maintained after the execution of lines 10--12 and 16 of \textsf{UpdateAC}.

Let $\widehat{G}_{r}$ be an affected auxiliary graph with $r \not= r_z$, and let $v$ and $u$ be vertices in $\widehat{G}_{r}$
such that $v$ is moved and $u$ is not.
Then, $v$ has an affected ancestor $t$ in $\widehat{G}_{r}$. Let $w$ be a vertex reachable from $v$ in $\widehat{G}_{r}$.
Then $w$ is also reachable from $t$ in $\widehat{G}_{r}$. We argue that $w$ is moved.
Let $\pi$ be a path from $t$ to $w$ in $\widehat{G}_{r}$, and let $t'$ be the first vertex on $\pi$
that is an ancestor of $w$ in $D$. If $t'=t$ then $w$ is moved. Otherwise, by the parent property of $D$, $\mathit{depth}(t') \le \mathit{depth}(t)$.
The fact that $(p,q)$ (the first bridge on the path from $z$ to $y$) is not locally canceled and
Lemma \ref{lemma:insert-affected} imply that $t'$ is affected. So $w$ is moved in this case as well.
Therefore, all vertices in $\widehat{G}_{r}$ that are reachable from $v$ are moved.
This means that if $v$ is related to $u$ then it has level
at least $k(u)$. Similarly, an edge $e$ that is related to $u$ and is moved has level
at most $k(u)$.
The invariant holds for $u$, since the vertices that are moved have level at least equal to the level of $u$.
Hence, the invariant is maintained for $\widehat{G}_{r}$.
Now consider $\widehat{G}_{r_z}$. The vertices in $\widehat{G}_{r_z}$ do not change level.
This is true also for $f(p)$, since it remains a canonical vertex even if its component is merged with some components from other auxiliary graphs.
Similarly, the original edges in $\widehat{G}_{r_z}$ also do not change level.
Suppose now that $e$ is a shortcut edge $(f(p),w)$ that is deleted and reinserted
as a shortcut edge $(u,w)$, with $u \not= f(p)$. Then $u$ is a moved vertex, so it has level $k(f(p))$.
Hence, shortcut edges also do not change level. Notice also that $e$
remains related to all the canonical vertices in $\widehat{G}_{r_z}$ it was related before.
Indeed, since $u$ is a descendant of $p$, $\widehat{G}_{r_z}$ contains a path from $f(p)$ to $u$.
If before the move there was a path in $\widehat{G}_{r_z}$ from a vertex $t$ to $e$ or vice versa,
then such a path exits after the move as well.
This implies that the invariant holds for the vertices that were already in $\widehat{G}_{r_z}$
before the insertion of $(x,y)$. Finally, consider a moved canonical vertex $u$.
Let $v$ be a vertex related to $f(p)$ with level $k(v)<k(f(p))$.
Then, there is a path from $v$ to $f(p)$, so after the move of $u$, there is a path from
$v$ to $u$. This implies that the invariant holds for $u$, since it holds for $f(p)$.

Hence, we showed that the invariant is maintained after the execution of lines 10--12 and 16 of \textsf{UpdateAC}.
By the proof of  \cite[Lemma 4.2]{Bender:IncCycleDetection:TALG}, it is also maintained after the execution of
line 17, so overall, subroutine \textsf{UpdateAC} maintains the invariant.
Hence, as in  \cite{Bender:IncCycleDetection:TALG}, the maximum level of a vertex is $\min \{  m/{\Delta} , 2n/\sqrt{\Delta}\}$,
since for every level other than the maximum, there are at least $\Delta$
different edges and $\sqrt{\Delta}$ different components.

So the total time spent by the \textsf{BottomIncSCC} data structures, excluding initialization, is $O(\min \{ m^{1/2}, n^{2/3}\} m) = O(mn)$.
The bounds for \textsf{SCInc2ECB-B} and \textsf{Inc2ECB-B} follow.
\end{proof}

\begin{theorem}
\label{theorem:Inc2ECB-time}
We can maintain the $2$-edge-connected blocks of a digraph with $n$ vertices through a sequence of edge insertions in
$O(mn)$ time, where $m$ is the total number of edges in $G$ after all insertions.
\end{theorem}

\section{$2$-edge-connectivity queries}
\label{sec:queries}

Here we provide the details of how to use our incremental algorithms for maintaining the $2$-edge-connected blocks of Sections
\ref{sec:incremental-strongly-connected} and Section \ref{sec:incremental-general}, in order to answer the following two types of queries:
\begin{mylist}{(a)}
	\litem{(a)} Test if two query vertices $u$ and $v$ are $2$-edge-connected; if not, report a separating edge for $u$ and $v$.
	\litem{(b)} Report all the $2$-edge-connected blocks.
\end{mylist}

A separating edge $e$ for $u$ and $v$ is a strong bridge that is contained in all paths from $u$ to $v$, or in all paths from $v$ to $u$.

First, we consider queries of type (a).
By Lemma \ref{lemma:labeling-algorithm}, $u$ and $v$ are $2$-edge-connected if and only if they are in the same subtree in the 
bridge decomposition and they belong to same auxiliary component
with respect to both the forward and the reverse flow graphs, $G_s$ and $G^R_s$. That is, $r_u = r_v$ and $c_u = c_v$ in $G_s$, and
$r^R_u = r^R_v$ and $c^R_u = c^R_v$ in $G^R_s$.
Recall that we keep the auxiliary components in $G_s$ (and similarly in $G^R_s$) using a disjoint set union data structure~\cite{dsu:tarjan}.
Since we aim at constant time queries,
we use such a data structure that can support each $\mathit{find}$ operation in worst-case $O(1)$ time and
any sequence of $\mathit{unite}$ operations in total time $O(n \log n)$~\cite{setunion:tvl}.
This way, we can identify the canonical vertex of the auxiliary component containing a query vertex in constant time. Hence,
we can test if $u$ and $v$ are $2$-edge-connected also in constant time.
If $u$ and $v$ are not $2$-edge-connected, then we wish to report a corresponding separating edge also in constant time.
Suppose first that $r_u \not = r_v$.
Without loss of generality, assume that $r_u$ is not a descendant of $r_v$ in $D$.
By Lemma \ref{lemma:partition-paths}, the strong bridge $(d(r_v),r_v)$ is a separating edge for $u$ and $v$.
Now consider the case where $r_u = r_v$, but $c_u \not = c_v$.
This means $u$ and $v$ are not strongly connected in the induced subgraph $G[D(r_u)]$, and therefore, all paths from $c_u$ to $c_v$, or all paths from $c_v$ to $c_u$, use vertices not in $D(r_u)$.
Without loss of generality, assume that all paths from $c_u$ to $c_v$ contain a vertex $w \notin D(r_u)$.
By Lemma \ref{lemma:partition-paths}, all paths from $c_w$ to $c_v$ go through $(d(r_u),r_u)$.
Thus, $(d(r_u),r_u)$ is a separating edge for $c_u$ and $c_v$.
We can find a separating edge for $u$ and $v$ when $r^R_u \not = r^R_v$ or $c^R_u \not = c^R_v$ similarly.

We now turn to queries of type (b) and show to report all the $2$-edge-connected blocks in optimal $O(n)$ time.
For each vertex $v$ we create the label $label(v) = \langle r_x, c_x, r^R, c^R \rangle$, and we insert the pair $\langle label(v),v \rangle$ into a list $L$.
As above, each of the values $r_x$, $c_x$, $r^R$, and $c^R$ is available in $O(1)$ time.
Next, we sort the list $L$ lexicographically in $O(n)$ time using bucket sorting.
In the sorted list $L$ the vertices of the same $2$-edge-connected block appear consecutively, since they have the same label.
Thus, all the $2$-edge-connected blocks can be reported in $O(n)$ time.


\begin{thebibliography}{10}

\bibitem{AW14}
A.~Abboud and V.~Vassilevska Williams.
\newblock Popular conjectures imply strong lower bounds for dynamic problems.
\newblock In {\em Proc. 55th {IEEE} Symposium on Foundations of Computer
  Science, {FOCS}}, pages 434--443, 2014.

\bibitem{domin:ahlt}
S.~Alstrup, D.~Harel, P.~W. Lauridsen, and M.~Thorup.
\newblock Dominators in linear time.
\newblock {\em SIAM Journal on Computing}, 28(6):2117--32, 1999.

\bibitem{dynamicdominator:AL}
S.~Alstrup and P.~W. Lauridsen.
\newblock A simple dynamic algorithm for maintaining a dominator tree.
\newblock Technical Report 96-3, Department of Computer Science, University of
  Copenhagen, 1996.

\bibitem{Bender:IncCycleDetection:TALG}
M.~A. Bender, J.~T. Fineman, S.~Gilbert, and R.~E. Tarjan.
\newblock A new approach to incremental cycle detection and related problems.
\newblock {\em ACM Transactions on Algorithms}, 12(2):14:1--14:22, December
  2015.

\bibitem{dominators:bgkrtw}
A.~L. Buchsbaum, L.~Georgiadis, H.~Kaplan, A.~Rogers, R.~E. Tarjan, and J.~R.
  Westbrook.
\newblock Linear-time algorithms for dominators and other path-evaluation
  problems.
\newblock {\em SIAM Journal on Computing}, 38(4):1533--1573, 2008.

\bibitem{clrs}
T.~H. Cormen, C.~E. Leiserson, R.~L. Rivest, and C.~Stein.
\newblock {\em Introduction to Algorithms, Second Edition}.
\newblock The MIT Press, 2001.

\bibitem{DI04}
C.~Demetrescu and G.~F. Italiano.
\newblock A new approach to dynamic all pairs shortest paths.
\newblock {\em J. {ACM}}, 51(6):968--992, 2004.

\bibitem{DI08}
C.~Demetrescu and G.~F. Italiano.
\newblock Mantaining dynamic matrices for fully dynamic transitive closure.
\newblock {\em Algorithmica}, 51(4):387--427, 2008.

\bibitem{EGIN97}
D.~Eppstein, Z.~Galil, G.~F. Italiano, and A.~Nissenzweig.
\newblock Sparsification -- {A} technique for speeding up dynamic graph
  algorithms.
\newblock {\em J. ACM}, 44(5):669--696, September 1997.

\bibitem{dominators:Fraczak2013}
W.~Fraczak, L.~Georgiadis, A.~Miller, and R.~E. Tarjan.
\newblock Finding dominators via disjoint set union.
\newblock {\em Journal of Discrete Algorithms}, 23:2--20, 2013.

\bibitem{FGN97}
P.~G. Franciosa, G.~Gambosi, and U.~Nanni.
\newblock The incremental maintenance of a depth-first-search tree in directed
  acyclic graphs.
\newblock {\em Inf. Process. Lett.}, 61(2):113--120, 1997.

\bibitem{F85}
G.~N. Frederickson.
\newblock Data structures for on-line updating of minimum spanning trees.
\newblock {\em SIAM J. Comput.}, 14:781--798, 1985.

\bibitem{Gabow:Poset:TALG}
H.~N. Gabow.
\newblock The minset-poset approach to representations of graph connectivity.
\newblock {\em ACM Transactions on Algorithms}, 12(2):24:1--24:73, February
  2016.

\bibitem{2ECB}
L.~Georgiadis, G.~F. Italiano, L.~Laura, and N.~Parotsidis.
\newblock 2-edge connectivity in directed graphs.
\newblock In {\em Proc. 26th ACM-SIAM Symp. on Discrete Algorithms}, pages
  1988--2005, 2015.

\bibitem{2VCB}
L.~Georgiadis, G.~F. Italiano, L.~Laura, and N.~Parotsidis.
\newblock 2-vertex connectivity in directed graphs.
\newblock In {\em Proc. 42nd Int'l. Coll. on Automata, Languages, and
  Programming}, pages 605--616, 2015.

\bibitem{dyndom:2012}
L.~Georgiadis, G.~F. Italiano, L.~Laura, and F.~Santaroni.
\newblock An experimental study of dynamic dominators.
\newblock In {\em Proc. 20th European Symposium on Algorithms}, pages 491--502,
  2012.

\bibitem{2C:GIP:arXiv}
L.~Georgiadis, G.~F. Italiano, and N.~Parotsidis.
\newblock {A New Framework for Strong Connectivity and 2-Connectivity in
  Directed Graphs}.
\newblock {\em ArXiv e-prints}, November 2015.

\bibitem{DomCert:TALG}
L.~Georgiadis and R.~E. Tarjan.
\newblock Dominator tree certification and divergent spanning trees.
\newblock {\em ACM Transactions on Algorithms}, 12(1):11:1--11:42, November
  2015.

\bibitem{Haeupler:IncTopOrder:TALG}
B.~Haeupler, T.~Kavitha, R.~Mathew, S.~Sen, and R.~E. Tarjan.
\newblock Incremental cycle detection, topological ordering, and strong
  component maintenance.
\newblock {\em ACM Transactions on Algorithms}, 8(1):3:1--3:33, January 2012.

\bibitem{2CC:HenzingerKL14}
M.~Henzinger, S.~Krinninger, and V.~Loitzenbauer.
\newblock Finding 2-edge and 2-vertex strongly connected components in
  quadratic time.
\newblock In {\em Proc. 42nd International Colloquium on Automata, Languages,
  and Programming (ICALP 2015)}, 2015.

\bibitem{HK95}
M.~R. Henzinger and V.~King.
\newblock Fully dynamic biconnectivity and transitive closure.
\newblock In {\em Proc. 36th {IEEE} Symposium on Foundations of Computer
  Science}, pages 664--672, 1995.

\bibitem{HK99}
M.~R. Henzinger and V.~King.
\newblock Randomized fully dynamic graph algorithms with polylogarithmic time
  per operation.
\newblock {\em Journal of the ACM}, 46(4):502--536, 1999.

\bibitem{HK01}
M.~R. Henzinger and V.~King.
\newblock Maintaining minimum spanning forests in dynamic graphs.
\newblock {\em SIAM J. Comput.}, 31(2):364--374, February 2002.

\bibitem{HLT01}
J.~Holm, K.~de~Lichtenberg, and M.~Thorup.
\newblock Poly-logarithmic deterministic fully-dynamic algorithms for
  connectivity, minimum spanning tree, 2-edge, and biconnectivity.
\newblock {\em J. ACM}, 48(4):723--760, July 2001.

\bibitem{Ita86}
G.~F. Italiano.
\newblock Amortized efficiency of a path retrieval data structure.
\newblock {\em Theor. Comput. Sci.}, 48(3):273--281, 1986.

\bibitem{Italiano2012}
G.~F. Italiano, L.~Laura, and F.~Santaroni.
\newblock Finding strong bridges and strong articulation points in linear time.
\newblock {\em Theoretical Computer Science}, 447:74--84, 2012.

\bibitem{2vcb:jaberi}
R.~Jaberi.
\newblock Computing the $2$-blocks of directed graphs.
\newblock {\em RAIRO-Theor. Inf. Appl.}, 49(2):93--119, 2015.

\bibitem{2VCC:Jaberi2015}
R.~Jaberi.
\newblock On computing the 2-vertex-connected components of directed graphs.
\newblock {\em Discrete Applied Mathematics}, 204:164--172, 2016.

\bibitem{King99}
V.~King.
\newblock Fully dynamic algorithms for maintaining all-pairs shortest paths and
  transitive closure in digraphs.
\newblock In {\em Proc. 40th {IEEE} Symposium on Foundations of Computer
  Science, {FOCS} '99}, pages 81--91, 1999.

\bibitem{domin:lt}
T.~Lengauer and R.~E. Tarjan.
\newblock A fast algorithm for finding dominators in a flowgraph.
\newblock {\em ACM Transactions on Programming Languages and Systems},
  1(1):121--41, 1979.

\bibitem{makino}
S.~Makino.
\newblock An algorithm for finding all the k-components of a digraph.
\newblock {\em International Journal of Computer Mathematics},
  24(3--4):213--221, 1988.

\bibitem{MarchettiSpaccamela:TopologicalSorting}
A.~Marchetti-Spaccamela, U.~Nanni, and H.~Rohnert.
\newblock Maintaining a topological order under edge insertions.
\newblock {\em Information Processing Letters}, 59(1):53 -- 58, 1996.

\bibitem{nagamochi}
H.~Nagamochi and T.~Watanabe.
\newblock Computing k-edge-connected components of a multigraph.
\newblock {\em IEICE Transactions on Fundamentals of Electronics,
  Communications and Computer Sciences}, E76--A.4:513--517, 1993.

\bibitem{PT07}
M.~P{\u{a}}tra{\c{s}}cu and M.~Thorup.
\newblock Planning for fast connectivity updates.
\newblock In {\em Proc. 48th IEEE Symposium on Foundations of Computer
  Science}, FOCS '07, pages 263--271, 2007.

\bibitem{irdom:rr94}
G.~Ramalingam and T.~Reps.
\newblock An incremental algorithm for maintaining the dominator tree of a
  reducible flowgraph.
\newblock In {\em Proceedings of the 21st ACM SIGPLAN-SIGACT symposium on
  Principles of programming languages}, pages 287--296, 1994.

\bibitem{dfs:t}
R.~E. Tarjan.
\newblock Depth-first search and linear graph algorithms.
\newblock {\em SIAM Journal on Computing}, 1(2):146--160, 1972.

\bibitem{domin:tarjan}
R.~E. Tarjan.
\newblock Finding dominators in directed graphs.
\newblock {\em SIAM Journal on Computing}, 3(1):62--89, 1974.

\bibitem{dsu:tarjan}
R.~E. Tarjan.
\newblock Efficiency of a good but not linear set union algorithm.
\newblock {\em Journal of the ACM}, 22(2):215--225, 1975.

\bibitem{st:t}
R.~E. Tarjan.
\newblock Edge-disjoint spanning trees and depth-first search.
\newblock {\em Acta Informatica}, 6(2):171--85, 1976.

\bibitem{setunion:tvl}
R.~E. Tarjan and J.~van Leeuwen.
\newblock Worst-case analysis of set union algorithms.
\newblock {\em Journal of the ACM}, 31(2):245--81, 1984.

\bibitem{Thorup2000}
M.~Thorup.
\newblock Near-optimal fully-dynamic graph connectivity.
\newblock In {\em Proc. 32nd ACM Symposium on Theory of Computing}, STOC '00,
  pages 343--350, 2000.

\bibitem{Thorup04}
M.~Thorup.
\newblock Fully-dynamic all-pairs shortest paths: Faster and allowing negative
  cycles.
\newblock In {\em Algorithm Theory - {SWAT} 2004, 9th Scandinavian Workshop on
  Algorithm Theory,}, pages 384--396, 2004.

\end{thebibliography}
\end{document}